\documentclass[10pt]{article} 
\usepackage[preprint]{tmlr}

\usepackage{hyperref}
\usepackage{url}

\usepackage{amssymb}
\usepackage{amsmath}
\usepackage{graphicx} 
\usepackage{amsthm} 
\usepackage{algorithm}
\usepackage{algorithmic}
\usepackage{natbib}
\usepackage{mathtools}
\usepackage[inline]{enumitem}
\usepackage{graphicx}
\usepackage{subcaption}
\usepackage{svg}
\usepackage{multirow}
\usepackage{booktabs}
\usepackage{placeins}

\title{ARMS: Automatic Reward Shaping for Sparse-Reward Multi-Agent Reinforcement Learning}

\author{\name Elie Abboud \email eliabboud1000@gmail.com \\
      \addr Department of Marine Technologies\\
      University of Haifa
      \AND
      \name Oren Gal \email orengal@univ.haifa.ac.il \\
      \addr Department of Marine Technologies\\
      University of Haifa}

\newcommand{\defeq}{\overset{\text{def}}{=}}
\newcommand{\Id}{\mathcal{I}} 
\newcommand{\Sd}{\mathcal{S}} 
\newcommand{\Ad}{\mathcal{A}} 
\newcommand{\Od}{\mathcal{O}} 
\newcommand{\Rd}{\mathcal{R}} 
\newcommand{\Td}{\mathcal{T}} 
\newcommand{\Dd}{\mathcal{D}} 
\newcommand{\ja}{\mathbf{a}} 
\newcommand{\jo}{\mathbf{o}} 
\newcommand{\jr}{\mathbf{r}} 
\newcommand{\tshaped}{\text{shaped}} 

\theoremstyle{plain}
\newtheorem{theorem}{Theorem}[section]

\newtheorem{corollary}[theorem]{Corollary}

\theoremstyle{definition}

\theoremstyle{remark}

\def\month{MM}  
\def\year{2026} 
\def\openreview{\url{https://openreview.net/forum?id=XXXX}} 

\newcommand\blfootnote[1]{%
  \begingroup
  \renewcommand\thefootnote{}\footnote{#1}%
  \addtocounter{footnote}{-1}%
  \endgroup
}

\begin{document}

\maketitle

\blfootnote{The authors used LLMs as an assistive tool for language editing, LateX troubleshooting and code generation. The authors carefully reviewed and tested the output.}

\begin{abstract}
Sparse rewards are a major bottleneck in multi-agent reinforcement learning (MARL), where simultaneous learning induces non-stationarity and makes reward design especially delicate. Reward shaping can accelerate learning, but in the multi-agent setting it must preserve the strategic structure of the problem rather than merely improve short-term optimization. We propose \textbf{A}utomatic \textbf{R}eward-shaping in \textbf{M}ulti-agent \textbf{S}ystems (\textbf{ARMS}), a self-supervised reward shaping framework for MARL that learns dense shaping signals from sparse environmental rewards through trajectory ranking. Since single-agent trajectory-ranking guarantees do not directly transfer to MARL, we reformulate policy invariance through conditional best-response reasoning, and show that if certain conditions hold, then using shaping rewards preserves each agent's best-response set under fixed opponent policies, and consequently preserve the set of Nash equilibria. Guided by this perspective, ARMS alternates between policy learning and reward learning while sharing shaping parameters across agents for effciency. Experiments in a partially observable multi-agent pathfinding domain show that ARMS improves sampling efficiency under increasing reward sparsity and agent count, generalizes to unseen environments, and reveals a MARL-specific failure mode in which limited exploration and coupled policy--reward dynamics induce oscillatory behavior. Increasing exploration mitigates this effect and stabilizes learning. To the best of our knowledge, ARMS is the first automatic reward shaping framework for MARL whose design is motivated by a game-theoretic equilibrium-preservation result.
\end{abstract}

\section{Introduction}

Multi-Agent Reinforcement Learning (MARL) has seen substantial progress, which is evident from the increasing number of papers and surveys in the field in recent years \citep{VDN2018, QMIX2020, MARL-Review2018, MARL-Review2021, MARL-Review2023, MARL-Communication-Survey2024, MARL-Review-Cooperative2023}. A central challenge in this setting is learning effectively from sparse rewards.

Sparse rewards are significantly more problematic in MARL due to non-stationarity and co-adaptation between agents, an element not present in single agent reinforcement learning (SARL). A common way to address sparse rewards is \emph{reward shaping}, namely augmenting the training signal with additional rewards that guide learning toward desirable behavior, or replacing the reward completely. Existing reward shaping methods are typically developed for single-agent settings, and extending them to MARL is non-trivial because the shaping signal must remain meaningful under coupled learning dynamics without altering the underlying strategic structure of the problem. In particular, reward shaping is risky because it may alter the set of optimal policies in SARL, or, in the multi-agent setting, the set of Nash equilibria~\citep{Nash1951}. For this reason, reward shaping methods that preserve Nash equilibria are especially attractive.

In reinforcement learning (RL), sparse reward functions are often easier to specify than dense ones, but they provide weak supervision and can substantially slow learning. Reward shaping aims to transform this sparse feedback into a denser training signal without changing the underlying solution of the problem. In MARL, this is especially delicate because multiple agents learn and act simultaneously in a shared environment, continually adapting their policies to one another as learning progresses. This gives rise to the inherent MARL challenge of non-stationarity \citep[Section~5.4.1]{marl-book}, under which the environment appears non-stationary from the perspective of each individual agent. Moreover, unlike SARL where the relevant solution is an optimal policy, in MARL the relevant solution is defined through mutual best responses and Nash equilibria: no agent can increase its expected return\footnote{The return is the cumulative reward obtained in a reinforcement learning problem.} by unilaterally altering its policy.

Designing such an informative dense reward is challenging even in single-agent settings \citep{ApprentLearning, SARLBook2018}. In multi-agent environments, this challenge becomes even harder, since the reward should encourage coordinated behavior between agents while preserving the set of Nash equilibria of the underlying problem \citep{MARL-Review-Cooperative2023}. This is precisely why reward shaping is both appealing and riskier in MARL: it has the potential to accelerate learning, but should not induce a shaping signal which changes the agents' objectives \textit{across all agents}.

Traditionally, \textit{reward shaping} refers to the hand-design of an additional reward signal by an expert, using domain knowledge or intuition to guide learning toward optimal policies. See \citet{PolicyInvariance1999} and \citet{PolicyInvariance-MultiAgent2011} for more information.

However, a more recent framework, called Self-Supervised Online Reward Shaping (SORS) \citep{SORS2021}, which removes the need for hand-designing reward was developed for SARL. The main idea is to use a parameterized method to infer a dense reward function by ranking trajectories based on the environment's sparse reward. Roughly speaking, the paper proposes a framework for training a parameterized reward function and shows theoretically that any two reward functions which similarly rank every pair of trajectories also induce the same set of optimal policies - which is the main motivation for inferring a reward function based on trajectory ranking.

In this paper, we propose \textbf{A}utomatic \textbf{R}eward-shaping in \textbf{M}ulti-agent \textbf{S}ystems (\textbf{ARMS}), a framework for automatic reward shaping in MARL (Algorithm \ref{alg:ARMS}, Section \ref{sec:method}). ARMS is built around a multi-agent view of trajectory-based reward learning, in which policy invariance is formulated through conditional best-response reasoning rather than global optimality (Section \ref{sec:NashInvariance}). The framework alternates between reinforcement learning and reward-learning steps, using sparse environmental rewards to supervise trajectory ranking. We further examine how multi-agent non-stationarity affects reward learning and training stability (Section~\ref{sec:agent-coordination-generalization} and~Section~\ref{sec:avoiding-suboptimal-policy-convergence}). To the best of our knowledge, ARMS is the first automatic reward shaping framework for MARL whose design is motivated by a game-theoretic equilibrium-preservation result.

While trajectory-ranking approaches are effective in single-agent reinforcement learning, they do not directly extend to multi-agent settings due to three fundamental challenges. First, in MARL, trajectory quality depends on the joint policy of all agents, making it impossible to evaluate or rank trajectories independently of other agents’ behaviors. Second, the learning process is inherently non-stationary, as agents simultaneously update their policies, causing the distribution of trajectories (and their induced rankings) to shift over time. Third, unlike the single-agent case, MARL lacks a global notion of optimality: there is no universally optimal policy or trajectory ordering, but rather a set of interdependent best-response solutions. As a result, the theoretical guarantees underlying SORS, which rely on consistent global trajectory rankings, no longer hold in this setting. These challenges necessitate a fundamentally different formulation of policy invariance for trajectory-based reward shaping in MARL.

Our contributions are threefold:
\begin{enumerate}
    \item \textbf{A multi-agent formulation of policy invariance for trajectory-based reward shaping.}
    The standard trajectory-ranking guarantees from the single-agent setting do not directly apply in MARL due to joint policy dependence and the absence of a global optimality notion. To address this, we reformulate policy invariance through conditional best-response reasoning, proving that total-order-equivalent shaping rewards preserve each agent's best-response set under fixed opponent policies, and consequently preserve the set of Nash equilibria. Section~\ref{sec:NashInvariance} presents this result.

    \item \textbf{ARMS: a self-supervised reward shaping framework for MARL.} 
    We propose ARMS, an online reward learning method for MARL that leverages trajectory ranking while accounting for the non-stationary and coupled nature of multi-agent learning. The framework learns a shared shaping signal across agents, and its design is motivated by the conditional best-response guarantee above. Section~\ref{sec:method} introduces the framework.

    \item \textbf{Empirical analysis of reward learning under multi-agent dynamics.} 
    We show that ARMS improves sampling efficiency as the number of agents increases and as rewards become sparser, while also generalizing to unseen environments. In addition, we identify a MARL-specific failure mode in which coupled policy--reward dynamics can induce oscillatory behavior under limited exploration, causing the shaping model to reinforce suboptimal cyclic behaviors. We show that increasing exploration mitigates this effect and stabilizes learning. Section~\ref{sec:experiments} details these experiments.
\end{enumerate}

In our experiments, we use a grid-world multi-agent pathfinding domain to enable interpretable analysis, since coordination behaviors are directly observable through agent interactions and solution quality is readily quantifiable. Each agent only observes local information, so the environment is modeled as a partially observable Markov game. Finally, we evaluate ARMS with two MARL backbones, IPPO and MAPPO, demonstrating that its benefits are not specific to a single base algorithm.

\subsection{Paper Organization}
The remainder of the paper is organized as follows: Section~\ref{sec:RelatedWorks} reviews related works, Section~\ref{sec:background} provides the necessary terminology and notation used in this work, specifically for multi-agent reinforcement learning (MARL) and Reward Shaping, Section~\ref{sec:NashInvariance} formalizes reward function equivalence and Nash Equilibria invariance in the multi-agent setting and establishes its role within the proposed framework, Section~\ref{sec:method} introduces the ARMS framework in detail and provides its pseudo-code implementation, Section~\ref{sec:experiments} presents empirical results demonstrating the effectiveness of ARMS against no-shaping and potential-based reward shaping (PBRS) in multiple sparse-reward settings across different RL backbones, and provides evidence that increased exploration mitigates premature convergence to suboptimal behaviors. Section~\ref{sec:future-work} concludes the paper and discusses future work.

\section{Related Works} \label{sec:RelatedWorks}

\subsection{Reward Shaping in SARL and MARL}

\textbf{In single agent reinforcement learning}, Reward shaping refers to the modification of the reward signal to accelerate learning. The modification is typically an incorporation of prior-knowledge into the task which reduces the number of suboptimal actions made to select the optimal action \citep{DriveBicycleShaping1998, PolicyInvariance1999}.

One way of doing so is by introducing an additive term to the original reward function. Let $\Rd(s,a,s') \in \mathbb{R}$ denote the reward function of the MDP underlying the RL problem, where the input tuple $(s,a,s')$ indicates the state-action-next state and let $F(s,a,s') \in \mathbb{R}$ denote a shaping term. Then it is replaced with
\[
\Rd'(s,a,s') \defeq \Rd(s,a,s') + F(s,a,s') 
\]

Reward shaping is risky because it can change the set of optimal policies. To deal with this problem potential-based reward shaping (PBRS) was proposed \citep{PolicyInvariance1999} which keeps the set of optimal policies unchanged. The form of the shaping term is simply:
\[
F(s,a,s')=\gamma \phi(s')-\phi(s)
\]
where $\gamma$ is the discount factor in the RL problem and $\phi(s) \in \mathbb{R}$ is any function.

\citet{SORS2021} proposes Self-Supervised Online Reward Shaping (SORS), which completely replaces the reward function with the shaped term:
\[
\Rd'(s,a,s')=F_\theta(s,a,s')
\]
where $F_\theta$ is the shaping term parameterized by $\theta$. The framework alternates between two phases: (a) a reinforcement learning phase, in which the agent is trained using the updated reward function, and (b) a reward update phase, in which the reward function is refined based on collected experiences during phase (a). The justification for completely replacing the original reward is that, if two reward functions define the same ordering over pairs of trajectories, then the two reward functions admit the same set of optimal policies (see Theorem \ref{thm:single-agent-total-order} for a formal statement).

\textbf{In the multi-agent setting}, each agent has its own policy\footnote{Generally, agents have distinct policies; parameter sharing uses a single shared parameterization across agents.} and thus instead of learning an optimal policy, a compromise must be made and so agents typically learn a Nash equilibrium \citep{Nash1951}. The result on potential-based reward shaping (PBRS) was extended to these settings and shown to not alter the Nash-Equilibrium \citep{MASPotentialShaping2011}.

\subsection{Learning a Reward Function from Preference Ranking}

Several prior works consider learning an inferred reward function from human preferences or ranking over trajectories, in the single-agent setting. 

One work, \citet{RewardShapingHumanPreferences2017}, fits a reward function based on human preferences. More concretely, they use a policy to generate pairs of trajectories and query the human on their preferences. The RL algorithm uses the updated reward network to train its agent. And \citet{RewardShapingAtari2018} build on the aforementioned work to use an initial set of expert trajectories to train an initial policy rather than start from a random policy. 

Another work, Self Supervised Online Reward Shaping (SORS) mentioned earlier, which uses an adaptation of the T-REX algorithm \citep{RewardShapingInverseRLTRex2019} for its reward inference part of the algorithm, learns a parameterized reward function online by ranking pairs of trajectories based on the original (sparse) reward. The reward network is then trained on the binary classification loss over these preferences. They show that their method collects higher reward than a baseline method throughout training (specifically Soft-Actor Critic) on a set of reward-sparse tasks.

Finally, \citet{Bui2025PreferenceGuidedLF} also addresses sparse-reward MARL by converting episodic outcomes into trajectory preferences and using them to learn an implicit reward model within a value-decomposition and dual-advantage MAPPO framework. In contrast, our method learns shaping rewards directly from trajectory ordering, and its central theoretical contribution is different: as shown in Section~\ref{sec:NashInvariance}, we establish a game-theoretic invariance result showing that, under total-order equivalence, best responses and therefore the set of Nash equilibria are preserved. This invariance result is the main theoretical motivation for our algorithmic design.

\section{Background and Definitions} \label{sec:background}

\subsection{Partially Observable Multi-Agent Reinforcement Learning}
We consider the Partially Observable Markov Decision Process (POMDP) in which multiple agents must make decisions simultaneously to maximize their ego-centric rewards: \newline
We denote the multi-agent reinforcement learning (multi-agent RL) with the tuple $M \defeq <\Id,\Sd,\Ad,\Rd,\Td,\Od,\gamma>$, in which:
\begin{itemize}
    \item $\Id=\{1,\ldots,N\}$ is the set of agents,
    \item $\Sd$ is the finite set of states of the environment.
    \item $\Ad=\{\Ad_1 \ldots,\Ad_N\}$ is the finite set of actions for all agents,
    \item A set of reward functions $\Rd\defeq\{\Rd_1,\ldots,\Rd_N\}$. Concretely, $\Rd_i:\Sd \times A \times \Sd \to \mathbb{R}$ is the reward function for agent $i$ where $A \defeq \Ad_1 \times \ldots \times \Ad_N$ where $\ja \in A$ denotes the \textbf{joint} actions across $N$ agents,
    \item $\Td:\Sd \times A \times \Sd \to [0,1]$ is the transition probability which defines the environment dynamics, for which $\Td(s,\ja,s')$ to be interpreted as the probability to transition from state $s$ to state $s'$ when the joint action $\ja$ was taken,
    \item A set of observation functions $\Od\defeq\{\Od_1,\ldots,\Od_N\}$. Concretely agent $i$'s observation function is given by $\Od_i:A \times \Sd \times \Od_i \to [0,1]$ where $\Od_i(s, \ja, o)$ denotes the probability that agent $i$ observes $o_i$ given that the environment was in state $s$ and the joint action $\ja$ was taken,
    \item and $\gamma\in[0,1]$ is the discount factor, which determines the relative importance of future rewards per agent \footnote{Specifically, rewards obtained $t$ time steps into the future are weighted by a factor of $\gamma^t$, thereby encouraging either short-term ($\gamma$ close to $0$) or long-term ($\gamma$ close to $1$) planning.}. We assume a common discount factor shared by all agents.
\end{itemize}

At each timestep $t \in \{1,\ldots,T\}$, where $T$ denotes the terminal timestep, the environment generates a state of the environment $s_t \in \Sd$, and each agent $i \in \Id$ receives its observation $o^i_t \in \Od_i$ with probability given by its observation function $\Od_i(\ja_{t-1},s_t,o^i_{t})$ and selects an action $a^i_t \in \Ad_i$ (agent $i$'s selection at timestep $t$). All agent selections at timestep $t$ form the joint action $\ja_t=(a^1_t,\ldots,a^N_t) \in A$. This joint action leads to a (stochastic) change in the environment generating a new state $s_{t+1}$ with probability $\Td(s_t,\ja_t,s_{t+1})$.  The environment also produces a \textit{joint reward} $\jr_t=(r_t^1,\ldots,r_t^N)\in\mathbb{R}^N$, where $r_t^i$ denotes the reward received by agent $i$ at timestep $t$.

Moreover, at each timestep $t$, each agent selects its action $a^i_t$ with probability given by its stochastic policy $\pi_i(a^i_t|h^i_t)$, where $h^i_t \defeq (o^i_0,\ldots,o^i_t)$, called the agent's \textit{history}, is all of the agent's observation up to timestep $t$.

The return for agent $i$ in a trajectory 
$\tau=(s_1,\ja_1,\jr_1,s_2,\ja_2,\jr_2,\ldots,s_T)$ is defined as
\[
\Rd_i(\tau) \defeq \sum_{t=1}^{T} \gamma^{t-1}\jr_t^i.
\]

Depending on the context, the state variables $s_t$ in $\tau$ may be replaced by the corresponding agent observations.

A \textbf{reward-free (joint) trajectory} is defined as $\tau\defeq(s_1,\ja_1,s_2,\ja_2,\ldots,s_T)$. When it is clear the trajectory doesn't include rewards, we just refer to $\tau$ as a (joint) trajectory, and define the joint observations of agents at timestep $t$ to be $\jo_t\defeq(o_t^1,\ldots,o_t^N)$.

To complete the definition of the POMDP, for any agent $i$ define $a^i_0=\emptyset$ and $o^i_0\defeq \emptyset$.

With a slight abuse of notation we sometimes use $\Td(s'|s,\ja)$ to refer to $\Td(s,\ja,s')$. This emphasizes that the selection of the next state is dependent on the previous action and previous state. In a similar fashion, we sometimes use $\Od_i(o^i_t|\ja_{t-1},s_t)$ instead of $\Od_i(\ja_{t-1},s_t,o^i_{t})$.

\subsection{Reward Shaping}

Given a MARL instance $<\Id,\Sd,\Ad,\Rd,\Td,\Od,\gamma>$, \textit{reward shaping} is the process of altering the original reward function for each agent to obtain a new (shaped) reward function either by replacing the original reward or augmenting it. 

More concretely, let $\Rd_i(s,\ja,s')$ be agent $i$'s reward function and let $F_{i,\theta}(s,\ja,s')$ be a shaping term for agent $i$ parametrized by a shared $\theta$ across agents. Reward shaping may then be defined either as an additive modification,
\[
\Rd_{i,\text{shaped}}(s,\ja,s') = \Rd_i(s,\ja,s') + F_{i,\theta}(s,\ja,s'),
\]
or as a replacement of the original reward,
\[
\Rd_{i,\text{shaped}}(s,\ja,s') = F_{i,\theta}(s,\ja,s').
\]

While training multiple agents in parallel it may encode the relationship between agent interactions implicitly, specifying how each agent should interact based solely on its local observations irrespective of its identity. Although $F_{i,\theta}$ is modeled as a function of the state-action-next state triple, in practice it will depend only on the local agent observations and actions. Finally, while the main goal of reward shaping is speeding up learning, using it may change the learned policy.

\section{Nash Equilibria Invariance} \label{sec:NashInvariance}

In a previous work, it is shown that in single-agent reinforcement learning it is possible to define a preference oracle over trajectories, which is a binary relation $\leq_p$ over trajectories. With a preference oracle $p$, one can rank trajectories:

$\tau_1\leq_p\ldots\leq_p\tau_k\leq_p\ldots$,
where the preference oracle can be defined by any deterministic reward function $\Rd$:

$\tau_i \leq_\Rd \tau_j \iff \Rd(\tau_i) \leq \Rd(\tau_j)$.

Define two deterministic reward functions $\Rd,\Rd'$ "total-order equivalent" if and only if for all trajectory pairs $\tau_i,\tau_j$:

$\tau_i \leq_{\Rd} \tau_j \iff \tau_i \leq_{\Rd'} \tau_j$.

Importantly, it is possible to show that if two deterministic reward functions $\Rd,\Rd '$ are total-order equivalent then they will induce the same set of optimal-policies. Let us formally restate the theorem:

\begin{theorem}[Single-Agent Total Order Equivalency (Restated from SORS  \citet{SORS2021})]\label{thm:single-agent-total-order}
Given a deterministic\footnote{That is, the transition dynamics for the MDP are deterministic.} reward-free (single agent) MDP $M=<\Sd,\Ad,\Td,\gamma>$ if two reward functions $\Rd, \Rd '$ are total order equivalent, they will induce the same set of optimal policies.
\end{theorem}

Next, we adapt this theorem to the MARL setting. The main idea of the proof is that if all agents policies are fixed except for agent $i$, then switching out that single agent's reward function $\Rd_i$ to a total-order equivalent $\Rd_i '$ induces the same set of optimal policies for that agent.

\begin{theorem}[Multi-Agent Total Order Equivalency]\label{thm:multi-agent-total-order}
Consider the deterministic reward-free multi-agent POMDP $M=<\Id,\Sd,\Ad,\Td,\Od,\gamma>$. Fix an agent $i$ and fix the policies of all other agents $\pi_{-i}$.

Let $M_i^{\pi_{-i}}$ denote the induced single-agent environment faced by agent $i$ when other agents follow $\pi_{-i}$. Let $\Rd, \Rd '$ be total-order-equivalent over trajectories from $M_i^{\pi_{-i}}$ then they induce the same set of optimal policies for agent $i$ in $M_i^{\pi_{-i}}$. Equivalently, the set of optimal policies for agent $i$ against $\pi_{-i}$ is invariant under replacing $\Rd$ with a total-order-equivalent $\Rd '$ deterministic reward functions.
\end{theorem}

\begin{proof}
The proof is structured by first showing that $M_i^{\pi_{-i}}$ is a deterministic single-agent MDP, and then apply Theorem \ref{thm:single-agent-total-order}. 

When other agent policies are fixed, agent $i$ is the only decision maker. Set $M_i^{\pi_{-i}}\defeq <\tilde \Sd,\tilde \Ad,\tilde \Td,\gamma>$, where $\tilde \Ad \defeq \Ad_i$ since agent $i$ is the only actor, the state representations at timestep $t$ denoted $\tilde s_t$ is simply the augmented state alongside the history of all agents $j\neq i$ up to timestep $t$: $\tilde s_t \defeq (s_t,h_t^{-i})$ - the histories are included to restore the Markov Property to the MDP since policies of agents $j \neq i$ may rely on histories to make an action selection, and finally $\tilde \Td$ simply transforms $(s_t,h_t^{-i})$ to $(s_{t+1},h_{t+1}^{-i})$ by computing all other agent actions under their fixed policy and advancing the state and histories deterministically. Observe that under these assumptions $\tilde \Td$ is deterministic.

Since $M_i^{\pi_{-i}}$ is a deterministic reward-free (single-agent) MDP and $\Rd, \Rd '$ are total-order-equivalent over trajectories from $M_i^{\pi_{-i}}$, then Theorem \ref{thm:single-agent-total-order} implies that $\Rd, \Rd '$ induce the same set of optimal policies in $M_i^{\pi_{-i}}$.  
\end{proof}

As an immediate consequence, preservation of each agent's best-response set implies preservation of the Nash equilibrium set.

\begin{corollary}[Nash Equilibrium Invariance]\label{cor:nash-invariance}
Suppose that for each agent $i \in \Id$ and for a fixed joint policy of other agents $\pi_{-i}$, the reward functions $\Rd_i$ and $\Rd_i'$ are total-order equivalent over trajectories in the induced single-agent environment $M_i^{\pi_{-i}}$. Then the induced games under $\{\Rd_i\}_{i\in\Id}$ and $\{\Rd_i'\}_{i\in\Id}$ have the same set of Nash equilibria.
\end{corollary}

\begin{proof}
By Theorem~\ref{thm:multi-agent-total-order}, for every agent $i$ and every fixed $\pi_{-i}$, the replacement $\Rd_i \mapsto \Rd_i'$ preserves the set of best responses. Hence each agent's best-response correspondence is unchanged. Since Nash equilibria are exactly the fixed points of these best-response correspondences, the set of Nash equilibria is unchanged.
\end{proof}

\subsection{Design Decisions and Applying Reward Function Equivalence in this work}

The single-agent notion of reward function equivalence does not transfer directly to MARL, since in the multi-agent setting the solution to the problem is no longer an optimal policy, but an equilibrium induced by mutual best responses. However, when the policies of the other agents are held fixed, each agent faces an induced single-agent problem. Under this conditional view, total-order-equivalent reward functions preserve each agent's best-response set, and, by Corollary~\ref{cor:nash-invariance}, preserve the set of Nash equilibria when the equivalence condition holds for all agents. This is the theoretical principle underlying our shaping design. Next, we describe the concrete design choices that operationalize this guarantee.

In this work, we use a parameterized reward function $F_\theta$, modeled as $F_\theta(s,\ja,s') \equiv F_\theta(o^i,a^i)$, where $o^i$ and $a^i$ are agent $i$'s observation and action. Each agent then defines its shaped reward by
\[
\Rd_{i,\tshaped}(s,\ja,s') \defeq F_\theta(o^i,a^i).
\]

While the theorem and corollary provide an equilibrium-preservation guarantee under exact total-order equivalence, in ARMS, we operationalize this through the reward-learning loss (See Equation \ref{eq:trajectory-ranking-loss}, Section \ref{sec:method}). Specifically, the loss encourages the learned shaping reward to be consistent with the preference ordering induced by the original reward. Since the loss is computed per agent, the reward function is optimized from an agent-wise perspective.

In our implementation all agents share the same policy network and the same reward-shaping parameters. By aggregating the resulting learning signal across agents, the shaping network is trained to infer rewards through agent-wise trajectory ordering. This shared and symmetric structure allows for more samples collected from various agents throughout training and makes scaling to larger agent counts easier, since a single network suffices regardless of the number of agents.

The choices above are specific to our implementation, but in principle, the reward shaper could process either each agent's observations separately or the joint observations of all agents. The former, which is our chosen setup, mirrors decentralized information flow, encourages reward primitives that generalize across agents, and, as noted above, increases the amount of training signal available to the reward-shaping objective, since each agent contributes its personal trajectory independently. The latter corresponds to a more centralized shaping design and is more suitable for learning a personalized reward per-agent or as a common reward for all agents\footnote{Common reward in MARL refers to the case where all agents receive the same reward scalar at each timestep.} which we leave as future work. Since the learned reward is used only during training, both choices are equally realistic to use; they differ primarily in the information structure used during reward-shaping optimization.

Finally, although the discussion above focuses on the design choices used in our implementation, Section~\ref{sec:method} presents the most general formulation of ARMS without assuming a specific information flow in the reward-shaper functions.

\section{Method} \label{sec:method}

We address the problem of sparse or delayed rewards in the multi-agent reinforcement learning (MARL) setting. Our key idea is to infer a dense reward functions for each agent that preserves the preference structure induced by the original sparse or delayed rewards. To this end, we introduce \emph{Automatic Reward-shaping in Multi-agent Systems} (ARMS), a general reward-shaping framework that is agnostic to the underlying reinforcement learning algorithm and can be seamlessly integrated into existing MARL pipelines.

Let $M = \langle \Id, \Sd, \Ad, \Rd=\{\Rd_1,\ldots,\Rd_N\}, \Td, \Od, \gamma \rangle$ denote a MARL problem with $N$ agents. Each agent $i \in \Id$ is associated with a parameterized shaped reward function $\Rd_{\theta_i}$, which is learned online during training.

\begin{algorithm}[t]
\caption{ARMS: Automatic Reward-shaping in Multi-agent Systems}
\label{alg:ARMS}
\let\tmlrAND\AND
\let\AND\relax
\begin{algorithmic}[1]
\REQUIRE MARL environment $M = \langle \Id, \Sd, \Ad, \{\Rd_1,\ldots,\Rd_N\}, \Td, \Od, \gamma \rangle$. 
\REQUIRE Length of each sampled trajectory $L$.
\REQUIRE Number of trajectory pairs to sample from buffer $K$.
\STATE \textbf{Notation:} $|\Id| = N$.
\STATE Initialize policy parameters $\{\pi_i\}_{i=1}^N$.
\STATE Initialize reward parameters $\{\theta_i\}_{i=1}^N$.
\STATE Initialize trajectory-pair buffer $\Dd \gets \emptyset$.

\WHILE{training not converged}

    \STATE // \textbf{Reinforcement Learning Phase:} \label{alg:ARMS:RL-Phase}
    \STATE Agents with policies $\{\pi_i\}_{i=1}^N$ interact with the environment using shaped rewards 
    $\{\Rd_{\theta_i}\}_{i=1}^N$, using any base RL algorithm for a fixed number of time-steps or episodes.
    \STATE Collect joint trajectories and store trajectory segments 
    $\tau$ in buffer $\Dd$.

    \STATE // \textbf{Reward-Shaping Phase:} \label{alg:ARMS:step-reward-shaping}
    \STATE Sample $K$ pairs of trajectory segments $(\tau_k, \tau_j)$ from $\Dd$, each of length $L$.

    \FOR{each agent $i \in \Id$}
        \STATE Update $\theta_i$ by minimizing loss $\mathcal{L}_i(\theta_i; \Dd)$ \hspace{1em} (Eq.~\ref{eq:trajectory-ranking-loss})
    \ENDFOR

\ENDWHILE

\RETURN Learned policies $\{\pi_i\}_{i=1}^N$ and shaped rewards $\{\Rd_{\theta_i}\}_{i=1}^N$
\end{algorithmic}
\let\AND\tmlrAND
\end{algorithm}

ARMS proceeds in two phases: a reinforcement learning phase and a reward-shaping phase. During the reinforcement learning phase, agents interact with the environment using their current shaped rewards $\Rd_{\theta_1},\ldots,\Rd_{\theta_N}$ and collect experience in the form of joint trajectories.

Specifically, let
\[
\tau_{t:t+L-1} = (\jo_t, \ja_t, \jo_{t+1}, \ja_{t+1}, \ldots, \jo_{t+L-1})
\]
denote a joint trajectory segment of length $L$, where $\jo_t = (o_t^1,\ldots,o_t^N)$ is the joint observation at timestep $t$. Collected trajectory segments are stored in a trajectory buffer $\Dd$.

During the reward-shaping phase, $K$ pairs of trajectory segments $(\tau_k,\tau_j)$ are sampled from $\Dd$. For each sampled pair, every agent $i$ performs a gradient update on its reward parameters $\theta_i$ using a binary classification loss that enforces consistency with the preference ordering induced by the original reward $\Rd_i$.

Formally, the loss used to train agent $i$'s shaped reward function $\Rd_{\theta_i}$ is defined as
\begin{equation}
\label{eq:trajectory-ranking-loss}
\begin{split}
\mathcal{L}_i(\theta_i; \Dd)
=
- \sum_{(\tau_k, \tau_j) \sim \Dd}
\Big[
&\mathbb{I}(\tau_k(i) \leq_{\Rd_i} \tau_j(i))\,
\log P(\tau_k(i) \prec \tau_j(i)) \\
&+
\big(1 - \mathbb{I}(\tau_k(i) \leq_{\Rd_i} \tau_j(i))\big)\,
\log P(\tau_k(i) \succ \tau_j(i))
\Big]
\end{split}
\end{equation}

where $\mathbb{I}(\cdot)$ is the indicator function, and $\tau_k(i),\tau_j(i)$ denote the portions of the joint trajectories corresponding to agent $i$'s observations and actions. This loss function has been used in other work to train a reward function with given pair-wise preference over trajectories \citep{SORS2021, RewardShapingHumanPreferences2017}.

The probability that trajectory $\tau_k$ is preferred over $\tau_j$ under the shaped reward is modeled using a softmax:
\begin{equation}
\label{eq:trajectory-preference-prob}
P(\tau_k \prec \tau_j)
=
\frac{\exp\!\left(\Rd_{\theta_i}(\tau_k)\right)}
{\exp\!\left(\Rd_{\theta_i}(\tau_k)\right)
+
\exp\!\left(\Rd_{\theta_i}(\tau_j)\right)}.
\end{equation}

Here, $\Rd_{\theta_i}(\tau)$ denotes the cumulative discounted shaped reward obtained by agent $i$ over trajectory $\tau$; for notational clarity we omit the agent index writing $\tau_k,\tau_j$ to denote $\tau_k(i),\tau_j(i)$.

\begin{figure}[t]
\centering
\includegraphics[width=\linewidth]{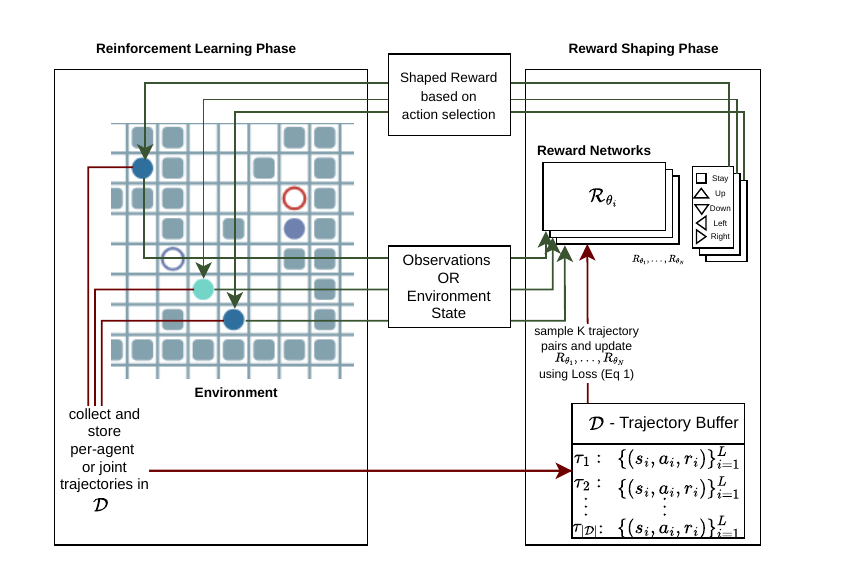}
\caption{Overview of the ARMS framework. Green arrows correspond to operations in the reinforcement learning phase, while red arrows correspond to reward-learning operations.}
\label{fig:ARMS-flow}
\end{figure}

The full algorithm is provided in Algorithm \ref{alg:ARMS} and Figure \ref{fig:ARMS-flow} shows the pipeline of ARMS schematically.

\section{Experiments}\label{sec:experiments}

\subsection{Experiment Setup}

We evaluate whether ARMS improves sampling efficiency and learning
under sparse reward feedback,
and whether it promotes coordinated behavior
in cooperative multi-agent reinforcement learning.
We test this in lifelong multi-agent path finding,
where agents must repeatedly reach assigned goals while avoiding collisions.
We compare ARMS with the base environment reward (in other words, with no reward shaping applied) and potential-based reward shaping
across two MARL backbones, multiple agent counts, and reward-sparsity levels,
and evaluate the generalization of the learned policies on 50 unseen maps.

\textbf{The problem.} We consider the following version of Lifelong Path-finding: several agents operate in a graph $G=(V,E)$ in which vertices correspond to locations and the edges to transitions between locations. Each agent starts in a given location and it needs to traverse the graph to reach its given goal location. When an agent reaches its goal, it is then assigned a new goal location. In every timestep, all agents perform their actions simultaneously. Each agent can either stay in its current location or move to an adjacent one. Agents must cooperate as to not collide with or block one another. In our experiments the graph is a grid, and episodes are of fixed length $T_{\max}$. During training we set $T_{\max} \coloneq 256$.

\textbf{The environment} used in our experiments is the scalable and fast POGEMA \citep{POGEMA}. One map is used in all our training experiments, a $20 \times 20$ map with $30\%$ obstacle density is generated using POGEMA's random map generator, Figure \ref{fig:Experiments-MAP} displays the map used. We evaluate the trained policy's sampling efficiency and quality on 50 maps not encountered during training - 40 random maps and 10 maze-like maps.

\begin{figure}[t]
\centering

\begin{subfigure}{0.48\textwidth}
    \centering
    \includegraphics[width=\linewidth]{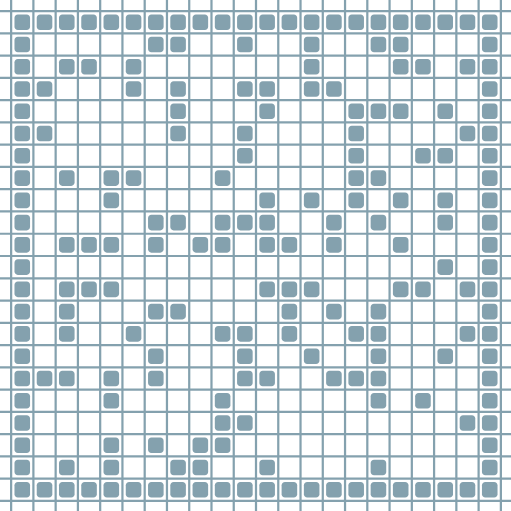}
    \caption{Empty map}
    \label{fig:Experiments-Map-no-agents}
\end{subfigure}
\hfill
\begin{subfigure}{0.48\textwidth}
    \centering
    \includegraphics[width=\linewidth]{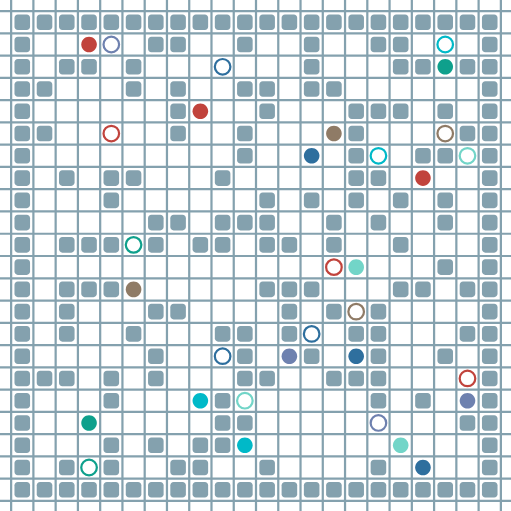}
    \caption{Map with 16 agents}
    \label{fig:Experiments-Map-16-agents}
\end{subfigure}

\caption{The map used in our experiments with $30\%$ obstacle density. Figure \ref{fig:Experiments-Map-no-agents} displays the empty map, while \ref{fig:Experiments-Map-16-agents} shows a typical episode initialization of 16 agents on the map; the filled circles denote agents while the hollow circles denote their corresponding targets.}
\label{fig:Experiments-MAP}
\end{figure}

\textbf{Agent observations} are based on a recent \citet{Follower}'s work. Each agent has a field of view (FOV) of its local environment. The view of radius $R \geq 0$ is the local $(2R+1)\times(2R+1)$ grid with the agent at the center. In our experiments we set $R \coloneq 5$, thus the FOV is of size $11 \times 11$. A variant of $A^*$ is used to generate a global path for each agent to its target location, of which is encoded within its FOV as way-points. More specifically the agent's observations are split into two matrices:
\begin{enumerate*}[label=(\arabic*)]
  \item \label{mat:obs+waypoints} One encoding shows the local obstacles within the FOV as $-1$ at that cell location, and the way-points along the computed path as $+1$,
  \item \label{mat:agents} The other, describing the local agents (including the self-agent placed at the center) as $+1$.
\end{enumerate*}

\textbf{The environment reward.} The reward is set to $r=+0.01$ each time the agent traverses into a way-point and $0$ otherwise, making for a dense reward function. However, our goal is to evaluate over a sparse-reward environment. For this reason, we change the reward to instead accumulate over several time-steps and give the agent the accumulated reward all at once; we use the terms \textit{original reward}, \textit{base reward}, and \textit{environment reward} interchangeably to refer to the reward provided by the environment, before adding any reward-shaping term. In the plots, this setting is denoted by \textit{No}, indicating that no reward shaping is applied. In all of our experiments we set the delay at $20$ time-steps, except in the second set of experiments in Section \ref{sec:experiments-training} in which we vary the delay period to show the robustness of our method across sparser rewards.

\textbf{The reward network structure.} The input to the network is the observations of a single agent. It is comprised of two parts:
\begin{enumerate*}[label=(\arabic*)]
  \item \label{rew-net:encoder} The Spatial encoder is a ResNet \citep{ResNet} with an additional Multi-Layer Perceptron (MLP) in the output layer of size $512$.
  \item \label{rew-net:MLP-head} Two additional MLP heads of size $128$ followed by a SiLU activation each.
  \item \label{rew-net:output-head} The output layer is an MLP head producing a $5$-dimensional vector, passed through a tanh activation, where each component represents the reward associated with selecting the corresponding action.
\end{enumerate*}
The training is end-to-end in all our experiments.

\textbf{Base RL algorithms.}
We use two multi-agent reinforcement learning backbones:
independent PPO (IPPO) and multi-agent PPO (MAPPO) \citep{PPO, IPPO, MAPPO}.
In both cases, the network follows the same architecture:
a spatial encoder first processes the local observation,
its output is passed to a GRU head modeling the observation history,
and the resulting representation is fed into actor and critic heads.
The actor outputs a distribution over the five available actions. During training, rollouts consisting of observations, actions, and rewards
are gathered asynchronously from multiple parallel environments with a fixed number of agents. 
The difference between IPPO and MAPPO is in the critic used for optimization:
IPPO trains agents using decentralized value estimates,
whereas MAPPO uses a centralized critic during training
while preserving decentralized execution.

\textbf{Reward-shaping variants.}
Throughout the experiments, we compare three reward settings.
The first, denoted "No" in the plots, trains the base MARL algorithm
directly on the sparse original environment reward.
The second, denoted "PBRS" (Potential Based Reward Shaping), augments the original reward with a hand-designed potential-based reward-shaping term, namely the commonly used manhattan distance in grid-world based environments.
The third, denoted \textbf{ARMS}, which is our reward-shaping signal.
In all cumulative-reward plots,
we report the accumulated original dense environment reward,
before applying reward sparsification.

\textbf{Training Plots.} Training commenced across $8$ parallel environments, each episode lasting $T_{\max} \coloneq 256$ timesteps. The cumulative reward plots in the following Sections show the average reward obtained across all agents throughout all $8$ parallel environments at that time-step, and they span several episodes.

\textbf{Evaluation.}
After training, the learned models are collected and evaluated
for episodes of length $T_{\max} \coloneq 256$.
We evaluate three aspects of the learned policies.
(1) First, we measure policy quality on the training map using average throughput,
defined as the number of times agents reach their goals divided by the episode length,
aggregated across all agents.
This metric captures the extent to which the learned policies solve
the lifelong path-finding task.
(2) Second, we evaluate whether the learned policies generalize beyond the training map
by measuring cumulative original dense environment reward
on the aforementioned 50 unseen maps.
(3) Finally, we measure agent coordination on the same 50 unseen maps
by reporting the normalized number of collisions accumulated during evaluation.
Sections~\ref{sec:avoiding-suboptimal-policy-convergence}
 and~\ref{sec:agent-coordination-generalization}
detail these results.

\textbf{Hyperparameters.} All the base-algorithm hyper-parameters relating to \citet{Follower} we keep unchanged, including the radius $R \coloneq 5$ mentioned earlier. The only exception is PPO's  exploration coefficient discussed in more detail in Section \ref{sec:agent-coordination-generalization}. The size and number of MLP heads in the reward network were optimized while holding all other parameters fixed, however the initial guess also outperformed the base algorithm. We additionally scale the reward from $[-1,+1]$ to $[-0.1,+0.1]$, which we observe to have greatly increased the performance of the network in terms of total accumulated original reward throughout training. \newline Little to no tuning was done to the hyper-parameters discussed next. First, we cap the trajectory buffer size at $16,384$ and set $K \coloneq 8,192$ which means $50\%$ of the buffer is renewed at every Reward Shaping Phase (see Algorithm \ref{alg:ARMS} step \ref{alg:ARMS:step-reward-shaping}). Each trajectory was of size $L \coloneq 16$ time-steps making several trajectories uninformative for the neural-network in the cases the reward delay was set to~$\in \{20, 30\}$. All experiments perform training for $\approx 4M$ agent-actions, thus when more agents are used, fewer timesteps (and therefore, fewer episodes) are present in the plots, for $16$ agents that is about $\approx 1K$ episodes.

\subsection{Sampling Efficiency} \label{sec:experiments-training}

In this section,
we evaluate whether ARMS improves sampling efficiency during training.
We compare ARMS against no reward shaping and potential-based reward shaping (PBRS),
using both IPPO and MAPPO as the underlying MARL backbone.
Unless stated otherwise,
the sparse environment reward is produced by accumulating the original dense reward
over $20$ timesteps and revealing the accumulated reward only at the end of the interval.
All cumulative-reward curves report the original dense environment reward,
before reward sparsification.

\begin{figure}[t]
\centering

\includegraphics[width=1.00\textwidth]{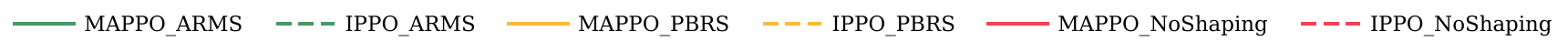}

\begin{subfigure}{0.48\textwidth}
    \centering
    \includegraphics[width=\linewidth]{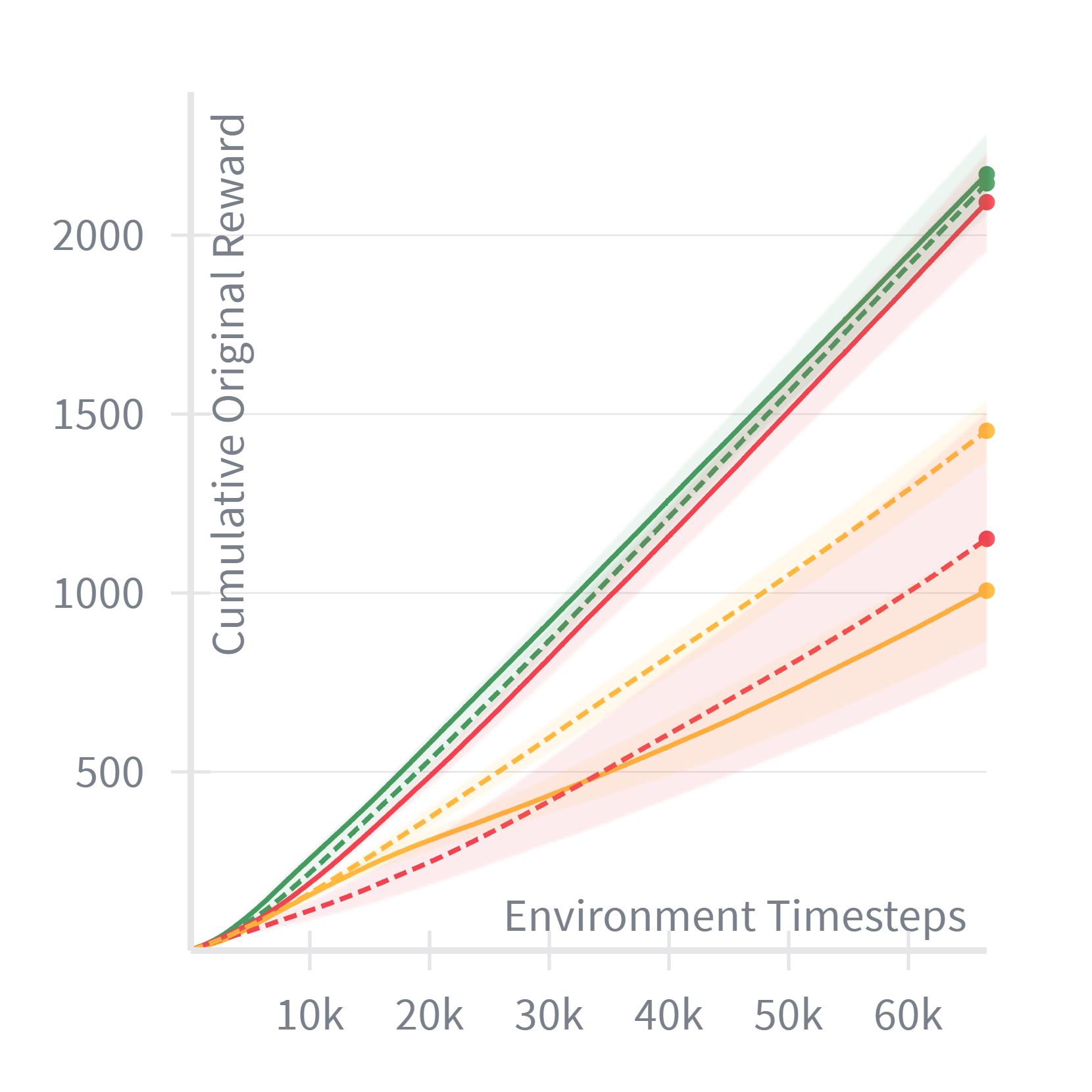}
    \caption{8 Agents}
\end{subfigure}
\hfill
\begin{subfigure}{0.48\textwidth}
    \centering
    \includegraphics[width=\linewidth]{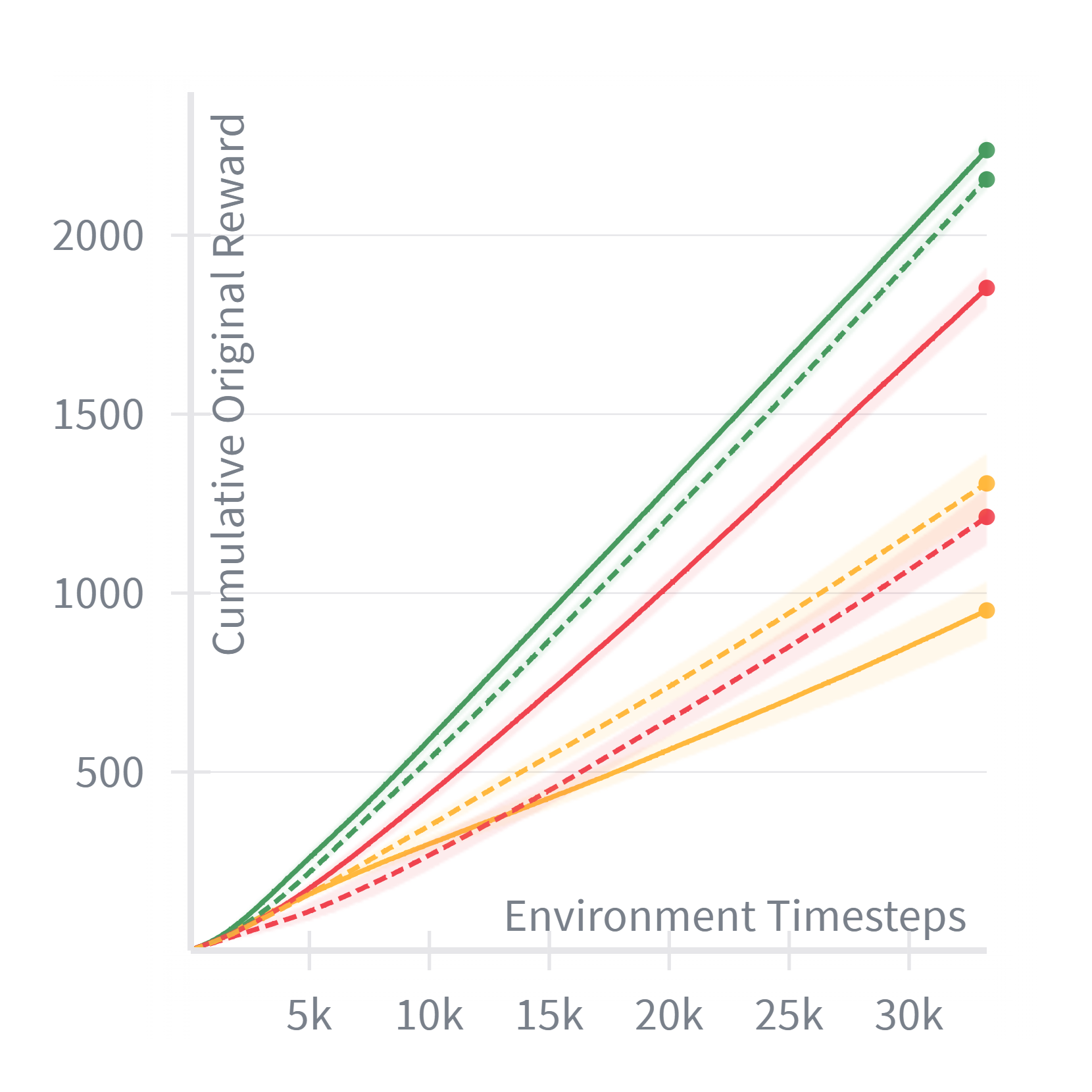}
    \caption{16 Agents}
\end{subfigure}

\vspace{0.6em}

\begin{subfigure}{0.48\textwidth}
    \centering
    \includegraphics[width=\linewidth]{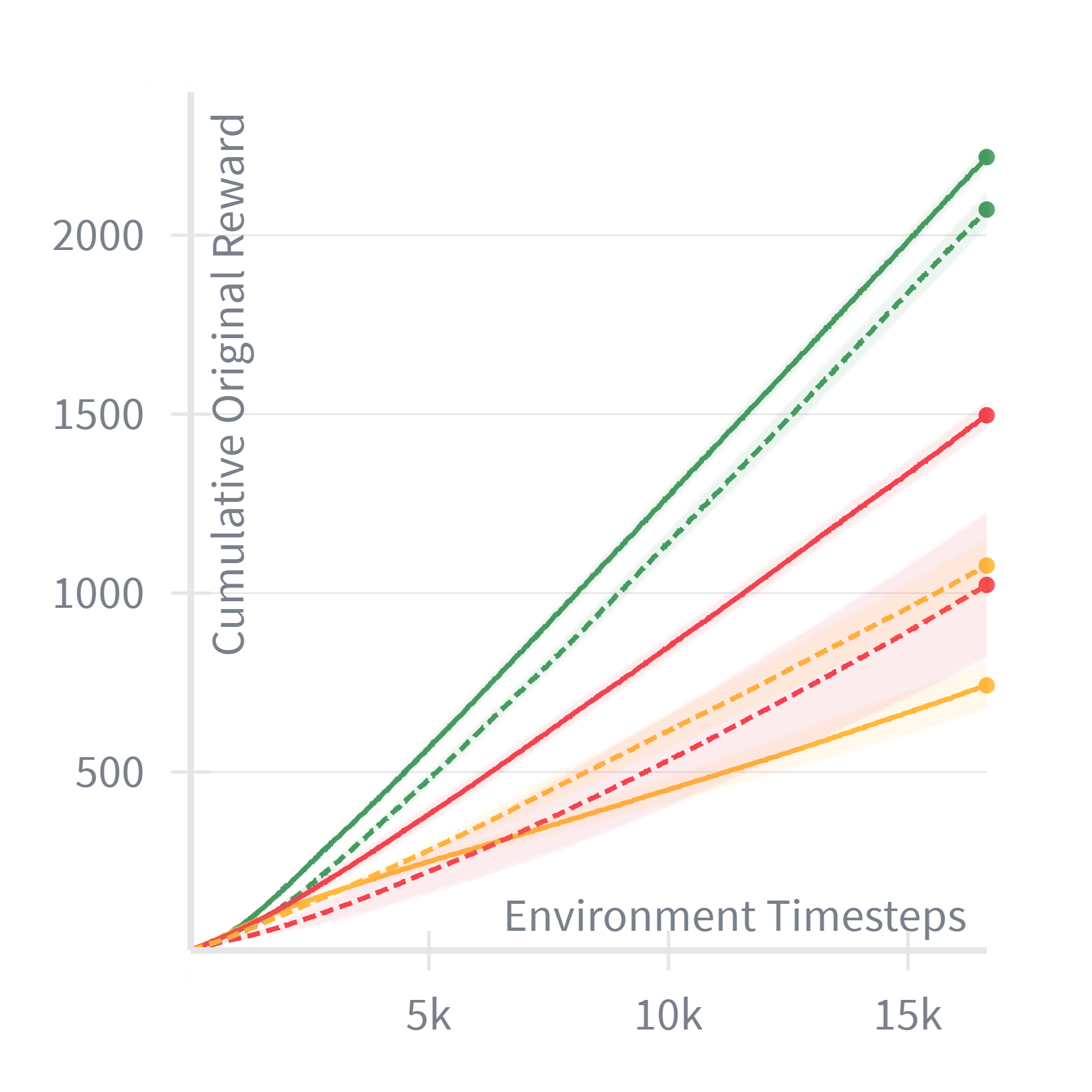}
    \caption{32 Agents}
\end{subfigure}
\hfill
\begin{subfigure}{0.48\textwidth}
    \centering
\end{subfigure}

\caption{Learning curves under a $20$-step sparse reward
for $8$, $16$, and $32$ agents.
We compare ARMS, PBRS, and no reward shaping,
each combined with IPPO and MAPPO.
The reported curves are the mean across $10$ seeds,
no curve smoothing is applied,
and the shaded region represents the standard deviation.}
\label{fig:vary-agents}
\end{figure}

In the first set of experiments,
we vary the number of agents while keeping the reward delay fixed at $20$ timesteps.
Figure~\ref{fig:Experiments-Map-16-agents}
shows an example episode initialization with $16$ agents,
and Figure~\ref{fig:vary-agents}
reports the corresponding learning curves for $8$, $16$, and $32$ agents.
Across these settings,
ARMS achieves higher cumulative original reward than PBRS and no shaping
in the more challenging multi-agent regimes.
For IPPO,
the improvement is strong across all agent counts.
For MAPPO,
the improvement is especially pronounced for $16$ and $32$ agents,
where coordination becomes increasingly important.
This provides early evidence that ARMS may be well-suited to induce inter-agent coordination in settings which have weak learning signals and harder multi-agent coordination demands,
a performance category further explored in
Section~\ref{sec:agent-coordination-generalization}.

\begin{figure}[t]
\centering

\includegraphics[width=1.00\textwidth]{legend.pdf}

\begin{subfigure}{0.48\textwidth}
    \centering
    \includegraphics[width=\linewidth]{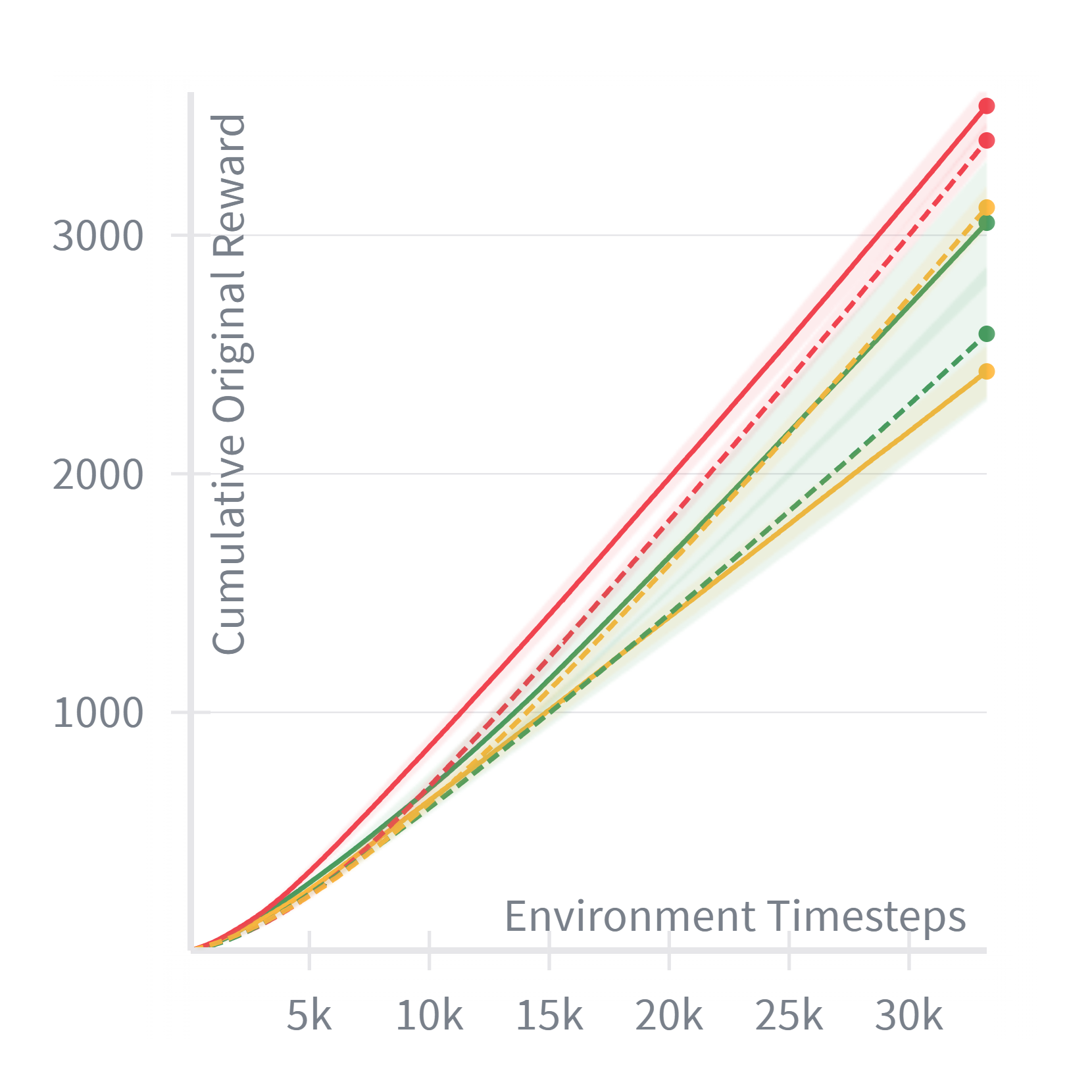}
    \caption{Dense reward (No delay)}
    \label{fig:16agents_rewards-a}
\end{subfigure}
\begin{subfigure}{0.48\textwidth}
    \centering
    \includegraphics[width=\linewidth]{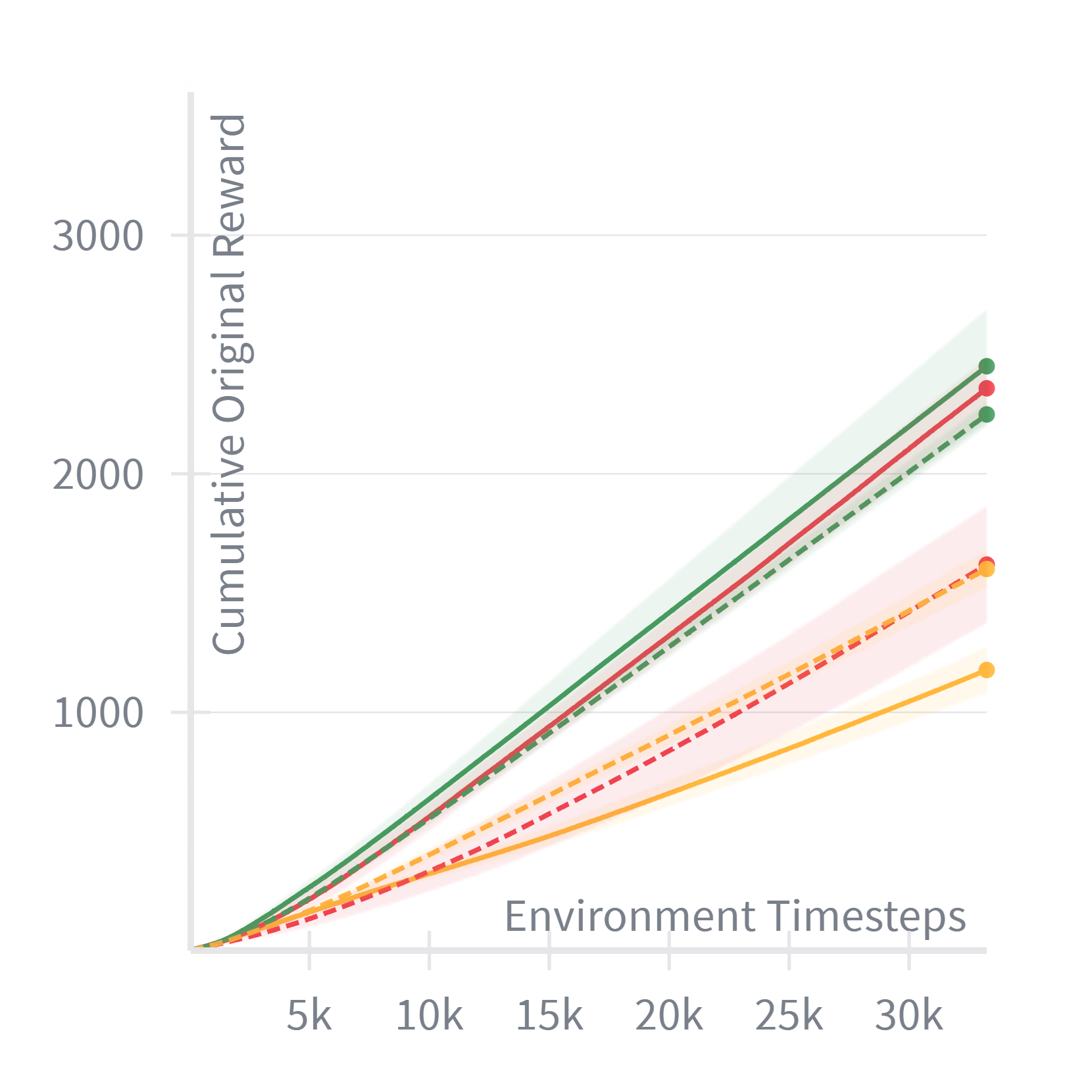}
    \caption{10-step sparse reward}
    \label{fig:16agents_rewards-b}
\end{subfigure}

\vspace{0.6em}

\begin{subfigure}{0.48\textwidth}
    \centering
    \includegraphics[width=\linewidth]{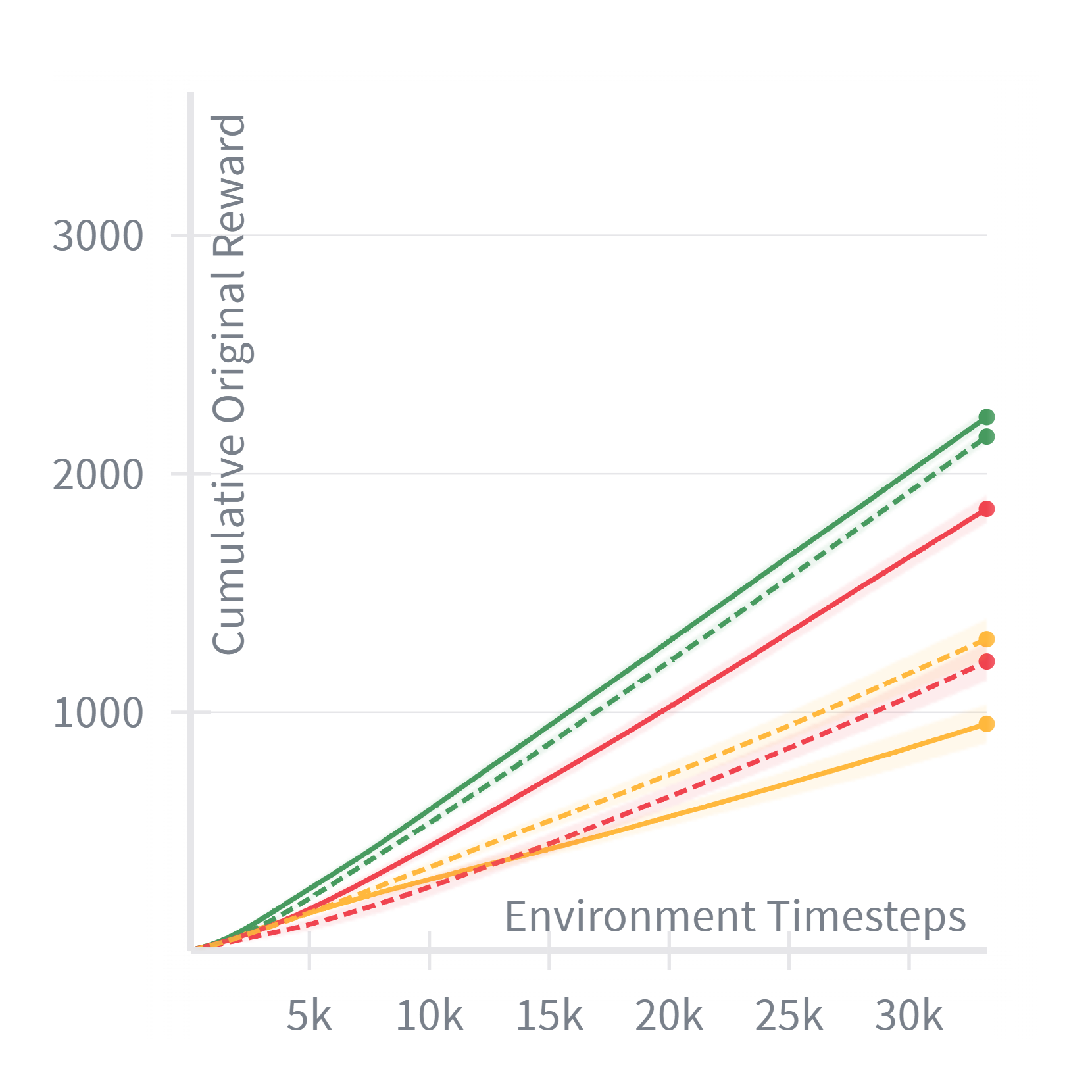}
    \caption{20-step sparse reward}
    \label{fig:16agents_rewards-c}
\end{subfigure}
\begin{subfigure}{0.48\textwidth}
    \centering
    \includegraphics[width=\linewidth]{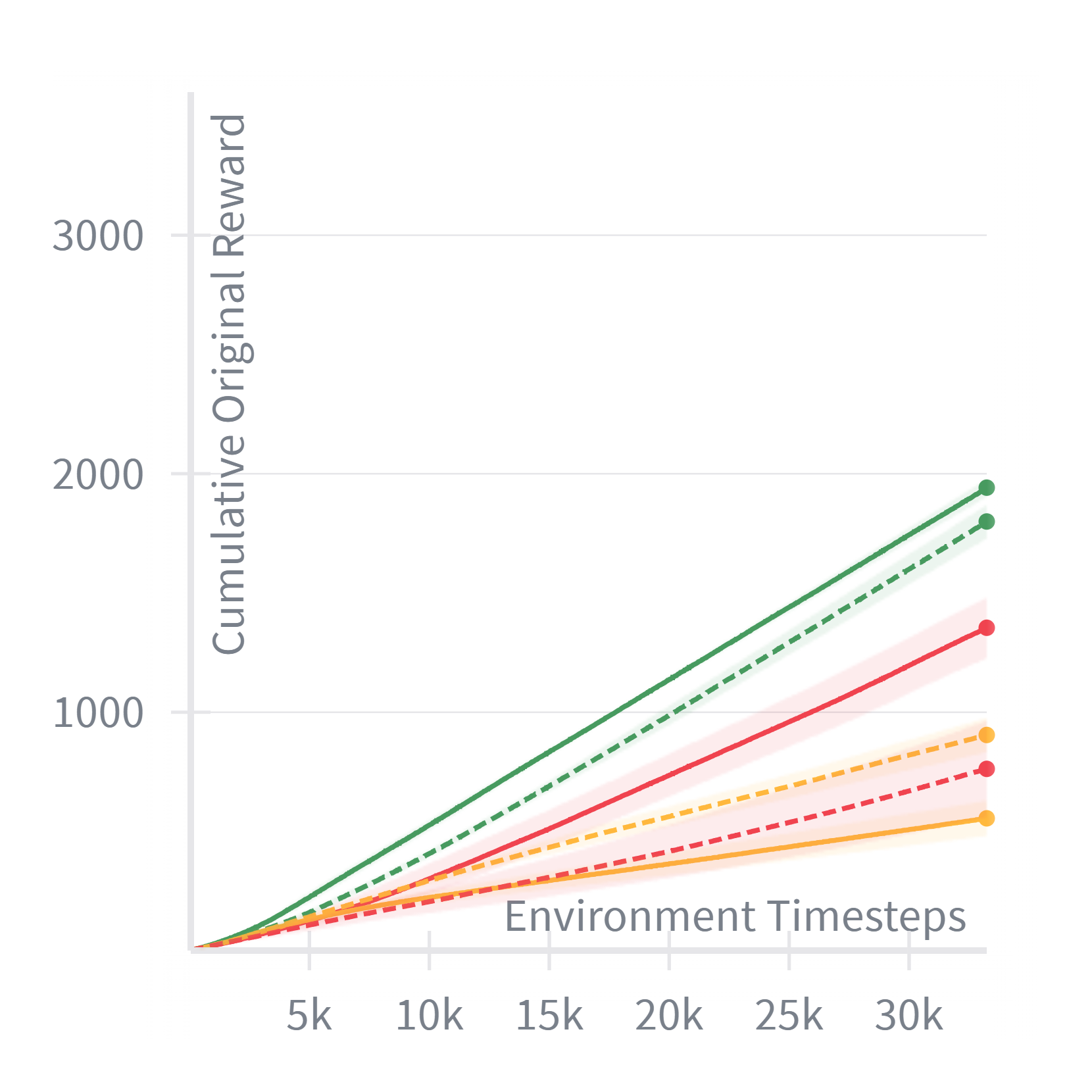}
    \caption{30-step sparse reward}
    \label{fig:16agents_rewards-d}
\end{subfigure}

\caption{Learning curves with $16$ agents under different reward-sparsity levels.
The dense setting corresponds to the original environment reward without delay.
The sparse settings accumulate the original reward over
$K \in \{10,20,30\}$ timesteps
and reveal it only at the end of the interval.
We compare ARMS, PBRS, and no reward shaping,
each combined with IPPO and MAPPO.
The reported curves are the mean across $10$ seeds,
no curve smoothing is applied,
and the shaded region represents the standard deviation.
}
\label{fig:16agents_rewards}
\end{figure}

Next,
we fix the number of agents at $16$
and vary the reward-sparsity level.
Figure~\ref{fig:16agents_rewards}
compares the dense original reward setting
against sparse reward settings with accumulation intervals
$K \in \{1\text{ (dense)}, 10,20,30\}$.
When the original reward is dense,
all methods receive frequent task feedback,
and ARMS does not provide an advantage, as expected.
As the reward becomes sparser,
ARMS provides a clearer benefit,
achieving substantially higher cumulative original reward
than PBRS and no shaping in the delayed-reward settings, while MAPPO is able to stay relatively competitive with ARMS in the $10$ sparse delay setting but performance begins to degrade as rewards become sparser. 
This supports the main motivation of ARMS:
learning a dense shaping signal is most useful
when the environment reward itself becomes less frequent and less informative.

This result is notable because the reward-network trajectories
have fixed length $L \coloneq 16$.
Thus,
when the reward delay is set to $20$ or $30$ timesteps,
some sampled trajectories may contain no revealed environment reward.
Nevertheless,
ARMS is still able to learn useful shaping signals
and improve training under these sparse-feedback regimes.

\begin{figure}[t]
\centering

\includegraphics[width=1.00\textwidth]{legend.pdf}

\begin{subfigure}{0.48\textwidth}
    \centering
    \includegraphics[width=\linewidth]{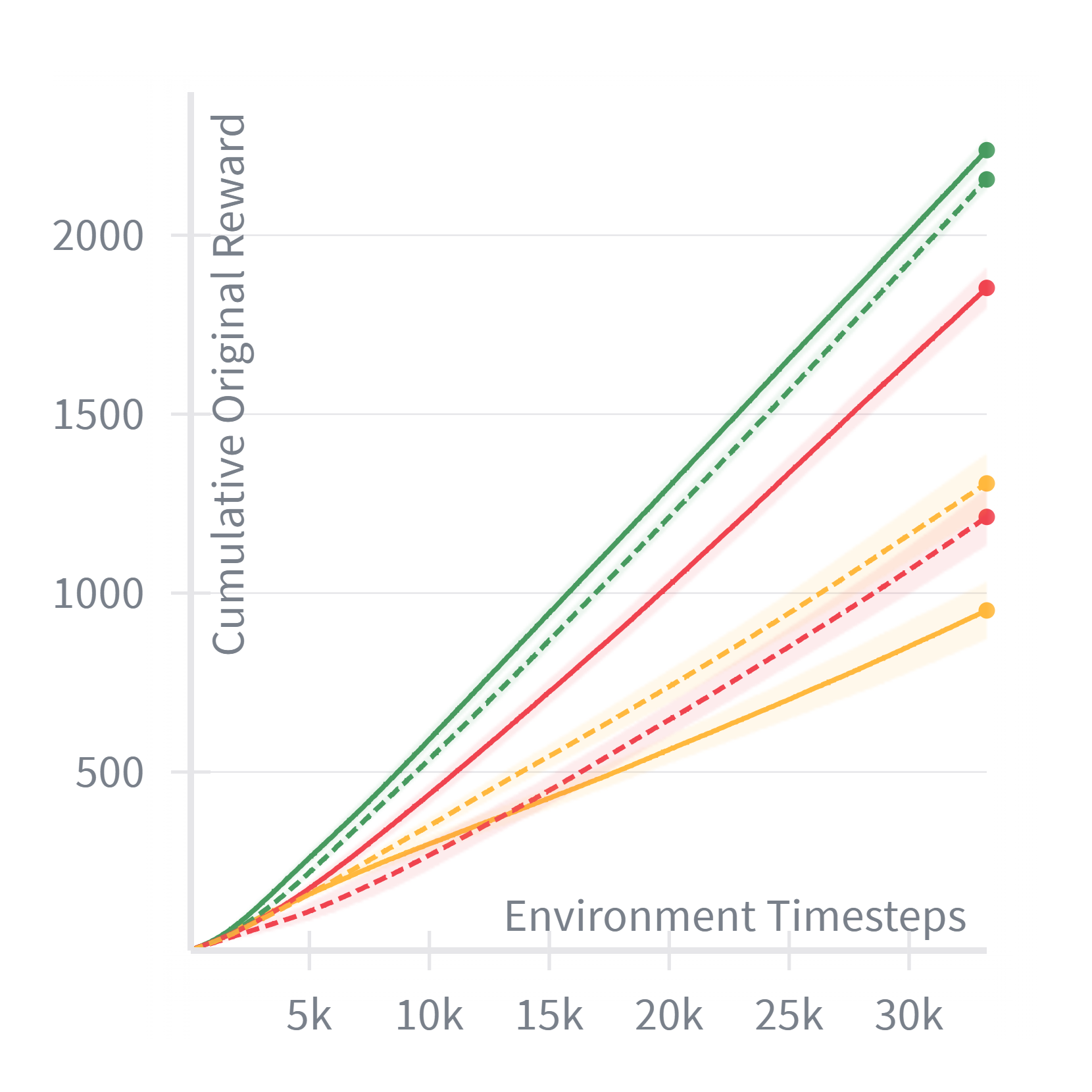}
    \caption{$\alpha = 0.023$}
    \label{fig:explore-0023}
\end{subfigure}
\hfill
\begin{subfigure}{0.48\textwidth}
    \centering
    \includegraphics[width=\linewidth]{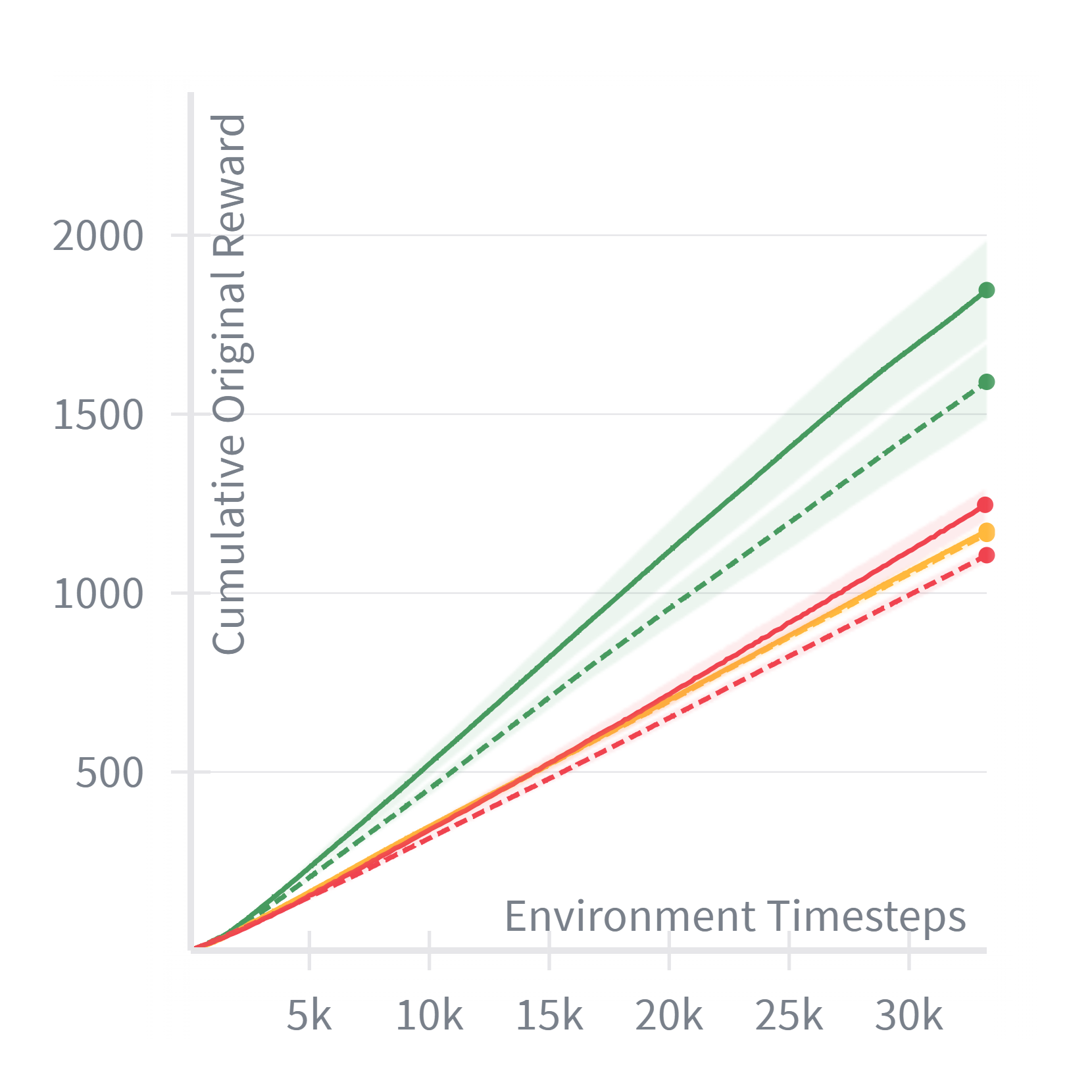}
    \caption{$\alpha = 0.15$}
    \label{fig:explore-015}
\end{subfigure}

\vspace{0.6em}

\begin{subfigure}{0.48\textwidth}
    \centering
    \includegraphics[width=\linewidth]{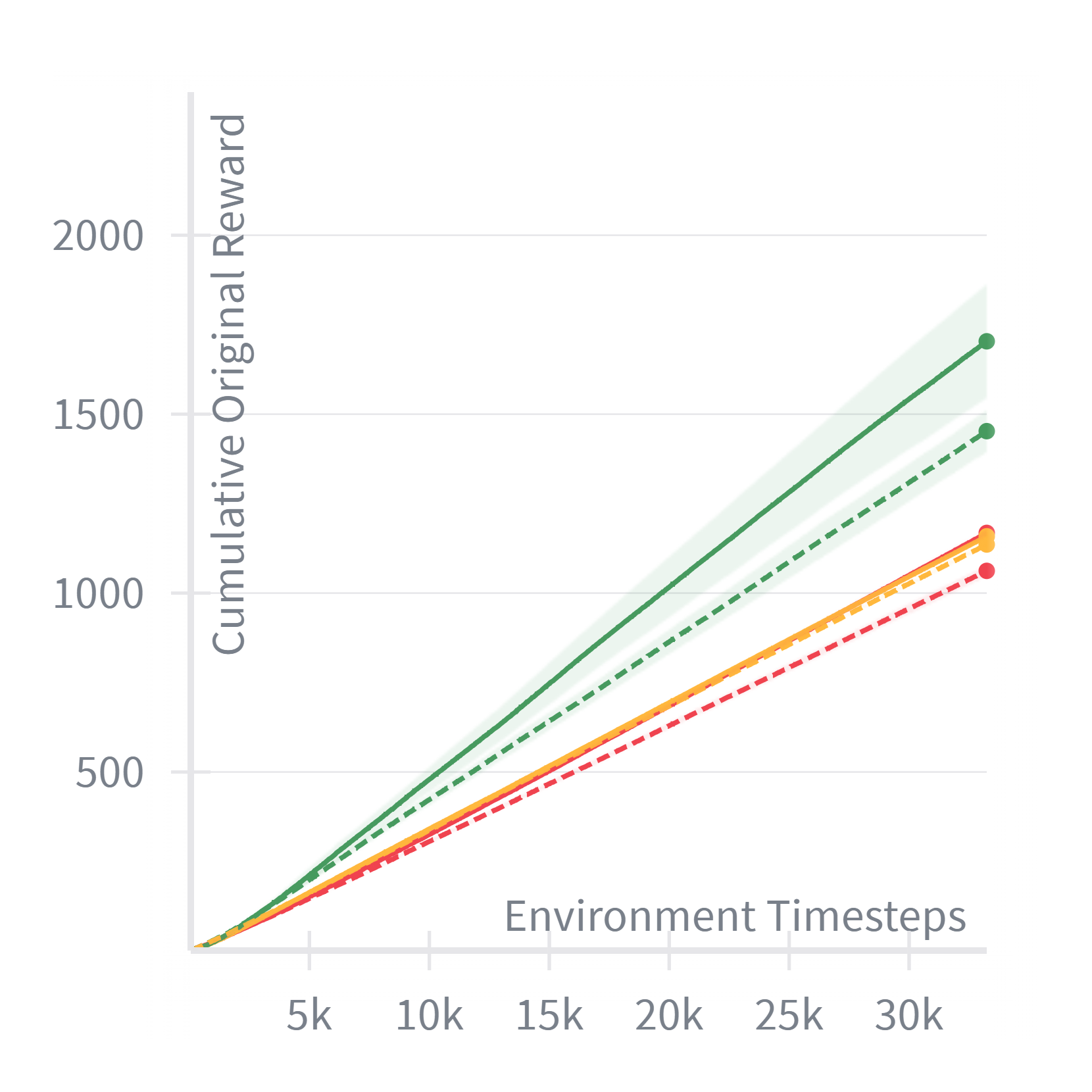}
    \caption{$\alpha = 0.25$}
    \label{fig:explore-025}
\end{subfigure}
\hfill
\begin{subfigure}{0.48\textwidth}
    \centering
    \includegraphics[width=\linewidth]{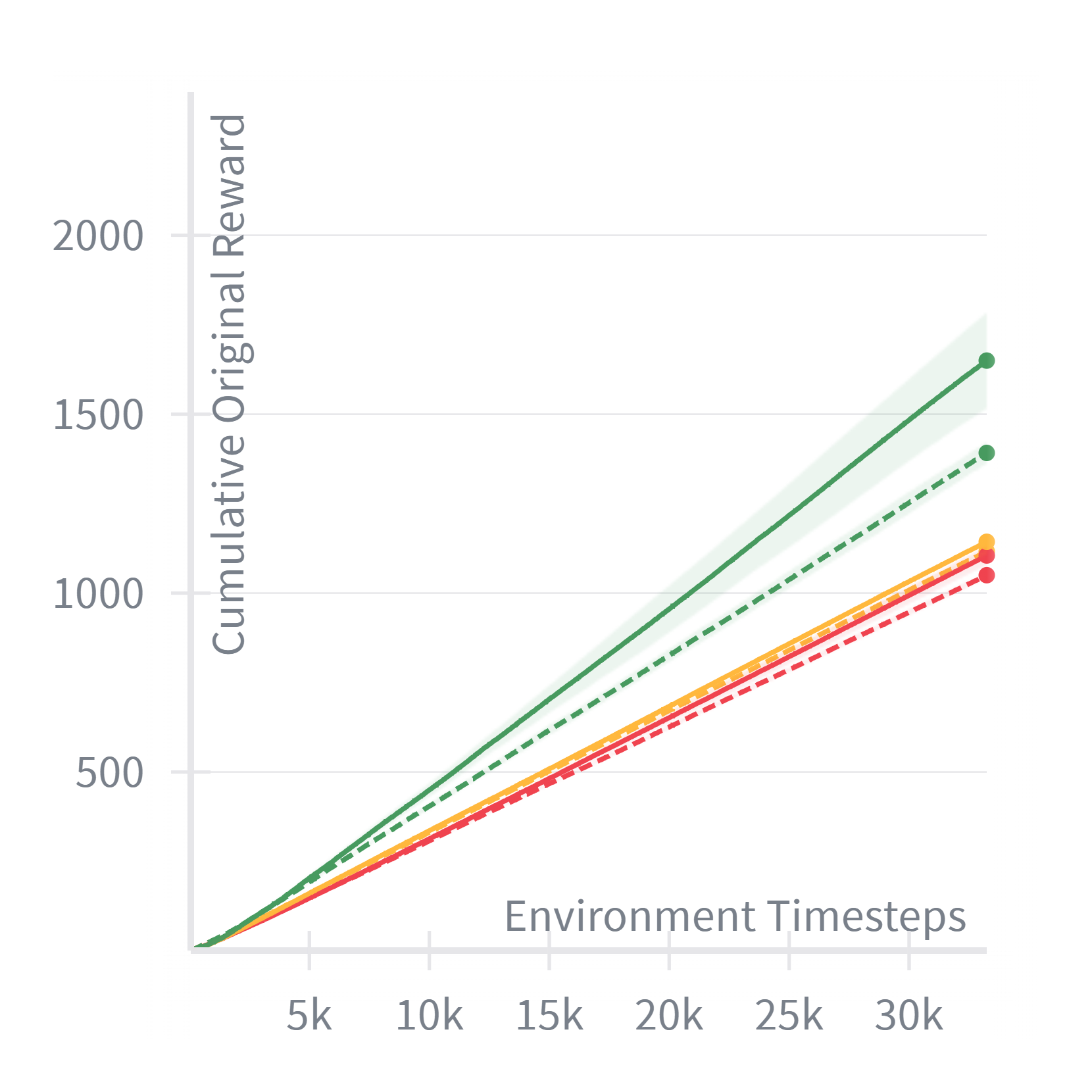}
    \caption{$\alpha = 0.35$}
    \label{fig:explore-035}
\end{subfigure}

\caption{
Effect of the exploration coefficient $\alpha$ on training performance.
Across all tested values of $\alpha$,
ARMS maintains higher cumulative original reward than PBRS and no shaping.
}
\label{fig:exploration-coeff-comparison}
\end{figure}

\subsection{Avoiding Suboptimal Policy Convergence} \label{sec:avoiding-suboptimal-policy-convergence}

In this set of experiments, we evaluate how to prevent ARMS into converging into suboptimal behaviour.

We show that varying the exploration coefficient has important consequences to the performance of the policy, which is especially difficult to learn in the MARL setting: agents continually co-adapt to each others behaviour and thus insufficient exploration can quickly lead to highly repetitive trajectories and stable but suboptimal coordination patterns. Particularily, in our chosen environment and underlying dense reward, agents can "hack" the reward by simply oscilitating repeatedly taking a positive small reward half the time and become unable to recover without enough exploration. This effect is especially pronounced in early training, agents will frequently collide or block one another during the RL phase (Algorithm \ref{alg:ARMS}, step \ref{alg:ARMS:RL-Phase}), producing trajectories that explore only a narrow region of the joint observation space. As a result the reward is repeatedly trained on similar interaction patterns during the reward-shaping phase (Algorithm \ref{alg:ARMS}, step \ref{alg:ARMS:step-reward-shaping}) and receives little signal about how alternative behaviors should be valued. Further, because agents share the policy parameters in our realization/implementation of ARMS, when agents block one another the policy is likely even more prone to this suboptimal behavior. As agents subsequently adapt their policies using this learned shaping reward (during the RL Phase), the system can enter a feedback loop where both the policy and the reward network specialize on this limited set of trajectories, making it difficult for the reward model to infer meaningful rewards for behaviors that have not yet ben observed. Increasing the exploration coefficient therefore encourages a more diverse trajectory distribution allowing the reward network to observe a wider range of interactions between agents and helping break this feedback loop. By measuring throughput we can measure whether the network has learned a meaningful policy with respect to solving the task.

In all previous experiments, the exploration coefficient was kept fixed at $0.023$. 
As shown in Figure~\ref{fig:throughput-arms-exploration}, agents trained with the ARMS framework under this setting achieve nearly zero throughput across MAPPO and IPPO, indicating convergence to suboptimal policies.

First, we verify that ARMS continues to dominate in terms of reward accumulation, so we repeat the first set of experiments from Section~\ref{sec:experiments-training} for $16$ agents while varying the exploration coefficient. 
As Figure \ref{fig:exploration-coeff-comparison} reports, ARMS consistently outperforms the baselines in reward accumulation across all tested exploration values, showing robustness to this parameter, in agreement with the earlier results reported in Figure~\ref{fig:vary-agents}. As $\alpha$ increases, the obtained reward in the final timestep may be lower but as we can see in Figure~\ref{fig:throughput-arms-exploration}, it mitigates this early convergence phenomenon.\footnote{The evaluated policies were the median performing policies in terms of cumulative reward at the end of training.}

\begin{figure}[h]
\centering
\includegraphics[width=0.62\textwidth]{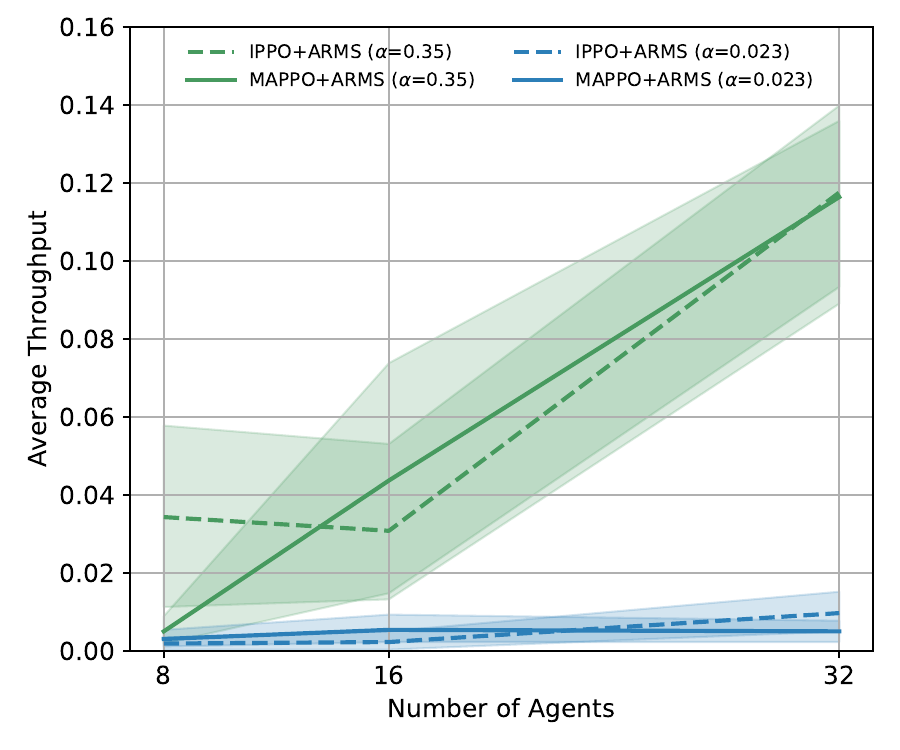}
\caption{
Average throughput of the learned ARMS policies
under low and high exploration.
We compare IPPO+ARMS and MAPPO+ARMS
with $\alpha=0.023$ and $\alpha=0.35$
across different agent counts.
The policies were evaluated across $100$ seeds,
and the shaded area represents the $95\%$ confidence interval.
}
\label{fig:throughput-arms-exploration}
\end{figure}

\begin{figure}[h]
\centering

\begin{subfigure}{0.48\textwidth}
    \centering
    \includegraphics[width=\linewidth]{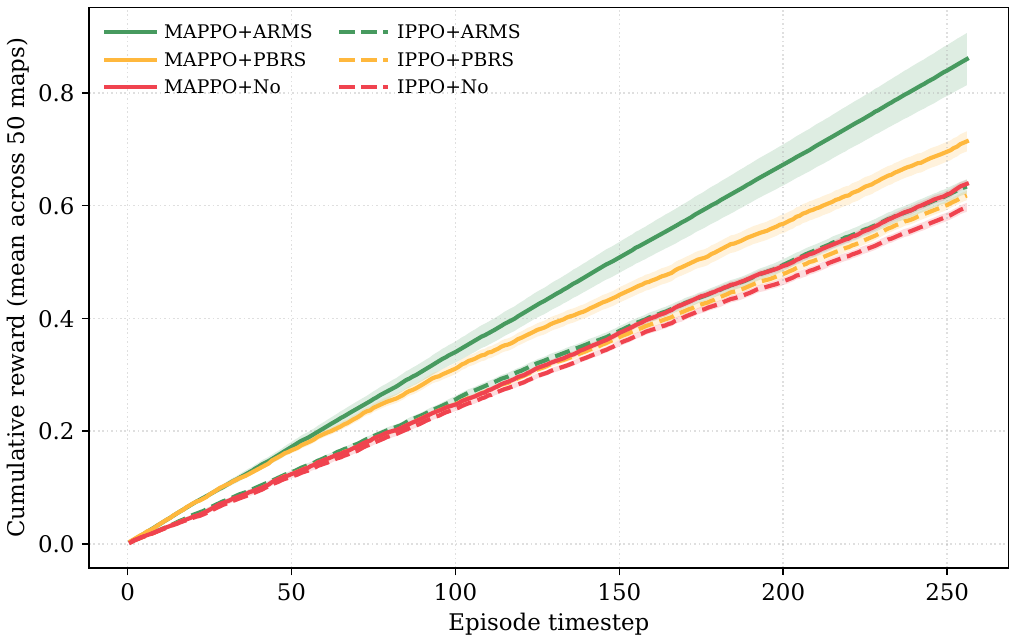}
    \caption{8 agents}
    \label{fig:eval-50maps-8agents}
\end{subfigure}
\hfill
\begin{subfigure}{0.48\textwidth}
    \centering
    \includegraphics[width=\linewidth]{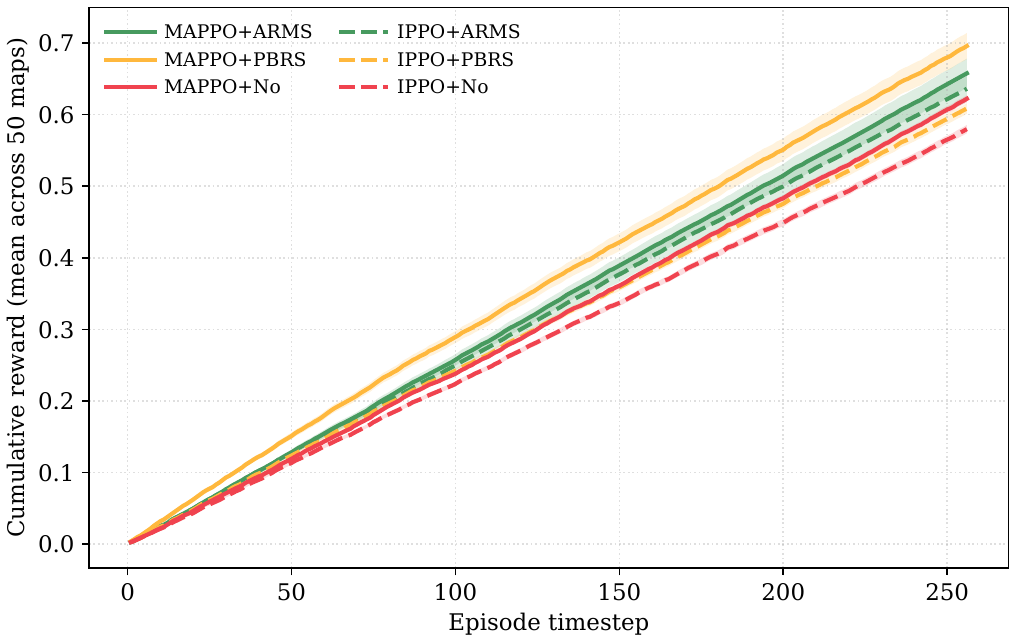}
    \caption{16 agents}
    \label{fig:eval-50maps-16agents}
\end{subfigure}

\vspace{0.6em}

\begin{subfigure}{0.48\textwidth}
    \centering
    \includegraphics[width=\linewidth]{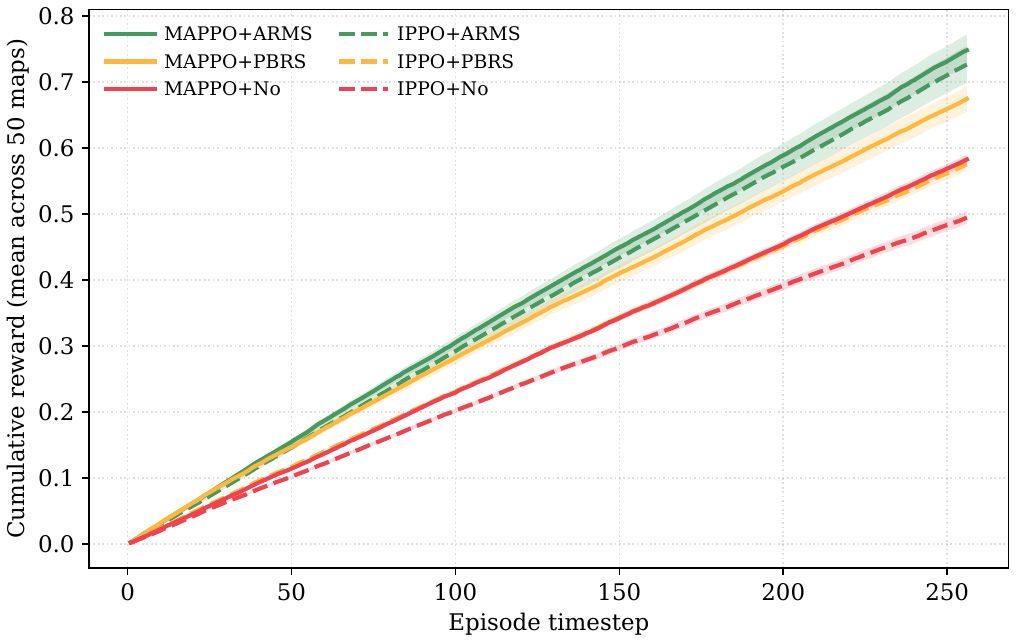}
    \caption{32 agents}
    \label{fig:eval-50maps-32agents}
\end{subfigure}
\hfill
\begin{subfigure}{0.48\textwidth}
    \centering
\end{subfigure}

\caption{
Cumulative original dense environment reward
on $50$ unseen evaluation maps
for $8$, $16$, and $32$ agents.
The evaluated policies were trained with exploration coefficient $\alpha=0.35$.
We compare ARMS, PBRS, and no reward shaping
with both IPPO and MAPPO.
}
\label{fig:eval-50maps-reward}
\end{figure}

\begin{figure}[h]
\centering

\begin{subfigure}{0.48\textwidth}
    \centering
    \includegraphics[width=\linewidth]{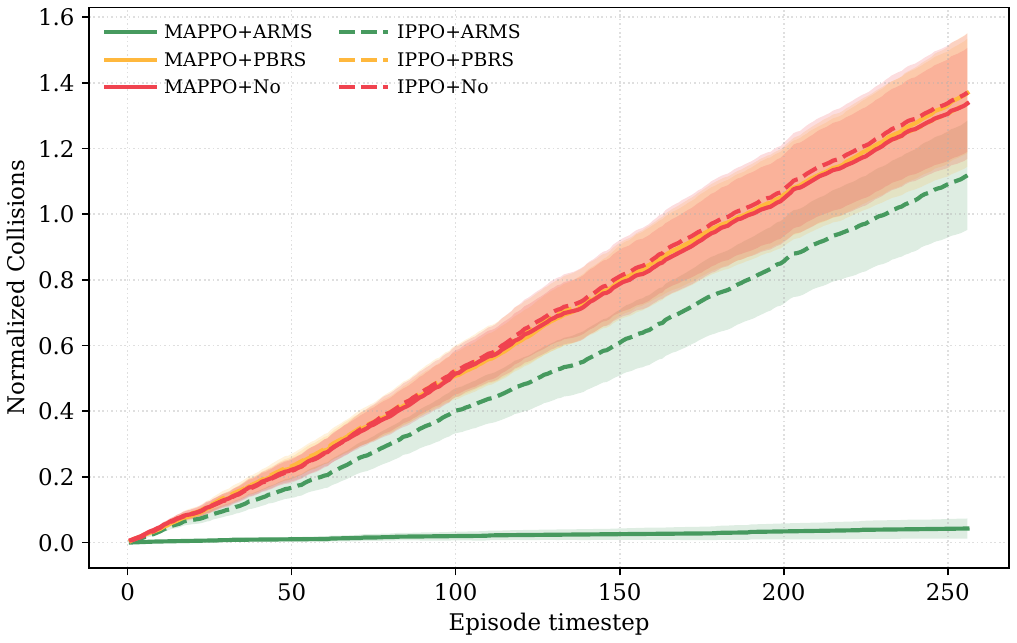}
    \caption{8 agents}
    \label{fig:collisions-50maps-8agents}
\end{subfigure}
\hfill
\begin{subfigure}{0.48\textwidth}
    \centering
    \includegraphics[width=\linewidth]{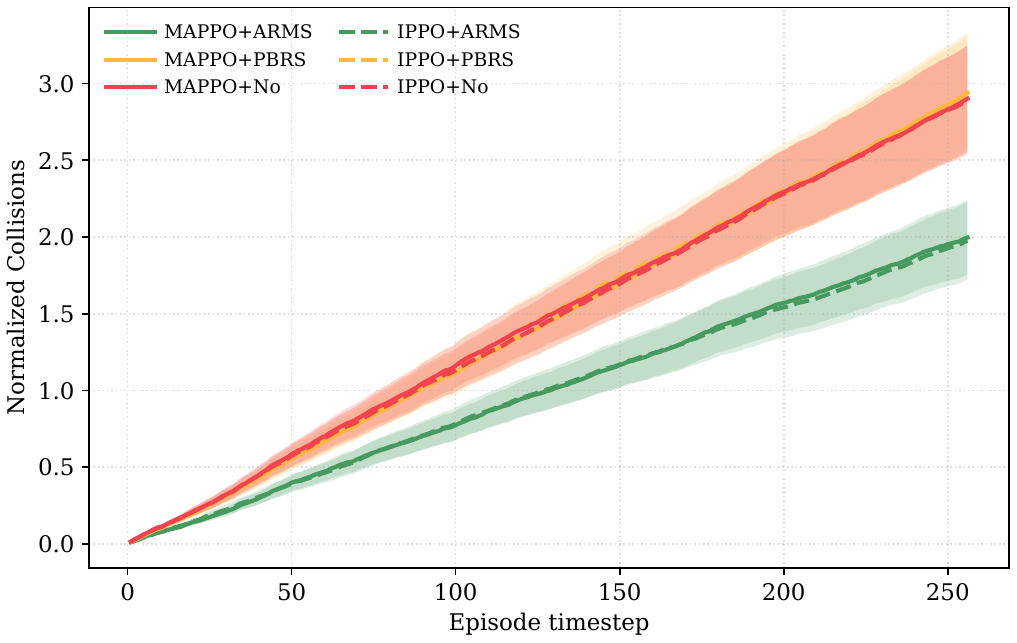}
    \caption{16 agents}
    \label{fig:collisions-50maps-16agents}
\end{subfigure}

\vspace{0.6em}

\begin{subfigure}{0.48\textwidth}
    \centering
    \includegraphics[width=\linewidth]{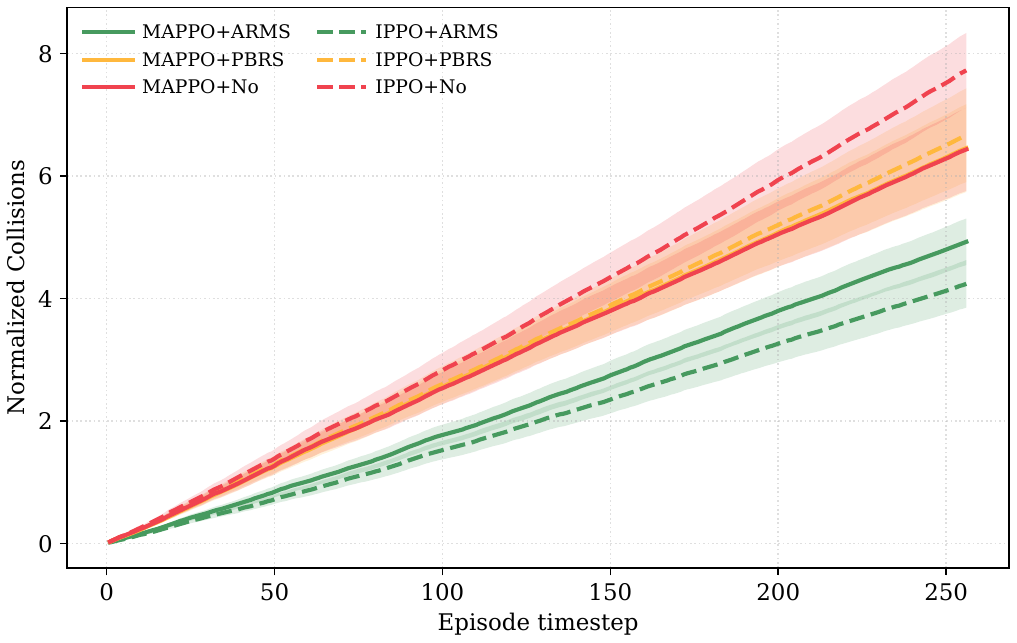}
    \caption{32 agents}
    \label{fig:collisions-50maps-32agents}
\end{subfigure}
\hfill
\begin{subfigure}{0.48\textwidth}
    \centering
\end{subfigure}

\caption{
Normalized collisions accumulated on $50$ unseen evaluation maps
for $8$, $16$, and $32$ agents.
The evaluated policies were trained with exploration coefficient $\alpha=0.35$.
We compare ARMS, PBRS, and no reward shaping
with both IPPO and MAPPO.
}
\label{fig:eval-50maps-collisions}
\end{figure}

\subsection{Agent coordination and generalization} \label{sec:agent-coordination-generalization}

In this set of experiments, we evaluate the generated policies, which were generated with the exploration coefficient set to $0.35$, on $50$ unseen maps during training - 40 random maps and 10 maze-like maps. First, we measure the accumulated original reward as we previously have across the episode, but also introduce another metric for the evaluation, \textit{normalized collisions}, defined as the number of total collisions to occur divided by the episode length. This measure helps us further determine whether the reward teaches a coordinated policy. Recall that the original dense reward gives a positive reward for every step taken along the path, thus when agents are not colliding as much, they should be accumulating reward much more quickly, which is precisely what ARMS tries to teach the agents as we will see.

Figure~\ref{fig:eval-50maps-reward} reports the cumulative original dense environment reward averaged over the $50$ unseen maps, each episode ran for $T_{max}\coloneq 256$ timesteps. Across the evaluated agent counts, ARMS generally accumulates reward more quickly than PBRS and no shaping, indicating that the learned shaping signal improves policy quality beyond the training map.

Figure~\ref{fig:eval-50maps-collisions}
reports the normalized collisions accumulated over the same $50$ unseen maps.
Across all evaluated agent counts,
ARMS produces substantially fewer collisions than PBRS and no shaping,
suggesting that the learned shaping signal improves coordination
rather than merely increasing reward accumulation.

We summarize the total accumulated dense reward throughout training for IPPO in Table~\ref{table:ippo_cumulative_reward}, the throughput in Table~\ref{table:arms_throughput_training}, and the normalized collisions for IPPO in Table~\ref{table:ippo_normalized_collisions}. Each entry reports mean~$\pm$~standard deviation reported in the earlier Figures.

\begin{table}[h]
\centering
\caption{Cumulative original dense (unshaped) reward on the training map for IPPO, broken down by (i) agent count at fixed reward sparsity~$=20$ and (ii) reward sparsity at fixed agent count~$=16$. Each cell reports mean$\pm$std observed in the final timestep. Columns compare the shaping methods No (unshaped), PBRS, and ARMS.}
\label{table:ippo_cumulative_reward}
\begin{tabular}{llccc}
\toprule
& & \multicolumn{3}{c}{Cumulative Original Reward} \\
\cmidrule(lr){3-5}
 & & No & PBRS & ARMS \\
\midrule
\multirow{3}{*}{Agent Count}
  &  8 & $1114.58 \pm 344.08$ & $1413.82 \pm 82.21$ & $2116.54 \pm 108.49$ \\
  & 16 & $1212.87 \pm  79.48$ & $1306.50 \pm 83.41$ & $2155.95 \pm  28.06$ \\
  & 32 & $1022.83 \pm 201.99$ & $1076.61 \pm 74.92$ & $2071.69 \pm  47.75$ \\
\midrule
\multirow{3}{*}{Reward Sparsity}
  & 10 & $1619.27 \pm 244.68$ & $1600.68 \pm 78.23$ & $2248.71 \pm 47.97$ \\
  & 20 & $1212.87 \pm  79.48$ & $1306.50 \pm 83.41$ & $2155.95 \pm 28.06$ \\
  & 30 & $ 762.55 \pm 200.84$ & $ 904.65 \pm 71.63$ & $1799.66 \pm 70.29$ \\
\bottomrule
\end{tabular}
\end{table}

\begin{table}[h]
\centering
\caption{Average throughput on the training map for ARMS reward shaping, broken down by algorithm (IPPO, MAPPO), agent count (8, 16, 32), and exploration coefficient ($\alpha\in\{0.023, 0.35\}$). Each cell reports mean$\pm$std over 10 evaluation seeds.}
\label{table:arms_throughput_training}
\begin{tabular}{llcc}
\toprule
Algorithm & Agent Count & $\alpha=0.023$ & $\alpha=0.35$ \\
\midrule
\multirow{3}{*}{IPPO}  &  8 & $0.002 \pm 0.003$ & $0.034 \pm 0.041$ \\
                       & 16 & $0.002 \pm 0.004$ & $0.031 \pm 0.036$ \\
                       & 32 & $0.010 \pm 0.009$ & $0.118 \pm 0.043$ \\
\midrule
\multirow{3}{*}{MAPPO} &  8 & $0.003 \pm 0.004$ & $0.005 \pm 0.006$ \\
                       & 16 & $0.006 \pm 0.006$ & $0.044 \pm 0.053$ \\
                       & 32 & $0.005 \pm 0.005$ & $0.116 \pm 0.037$ \\
\bottomrule
\end{tabular}
\end{table}

\begin{table}[h]
\centering
\caption{Normalized collisions at the final timestep $T_{max} \coloneq 256$ on the 50-map out-of-distribution benchmark, for IPPO with $\alpha=0.35$ and reward sparsity $20$. Each cell reports mean$\pm$std across the 50 evaluation maps.}
\label{table:ippo_normalized_collisions}
\begin{tabular}{lccc}
\toprule
Agent Count & No & PBRS & ARMS \\
\midrule
 8 & $1.371 \pm 0.690$ & $1.360 \pm 0.743$ & $1.083 \pm 0.594$ \\
16 & $2.875 \pm 1.312$ & $2.958 \pm 1.424$ & $1.955 \pm 0.902$ \\
32 & $7.770 \pm 2.231$ & $6.619 \pm 2.681$ & $4.273 \pm 1.448$ \\
\bottomrule
\end{tabular}
\end{table}

\begin{figure}[h]
    \centering
    \includegraphics[width=0.85\linewidth]{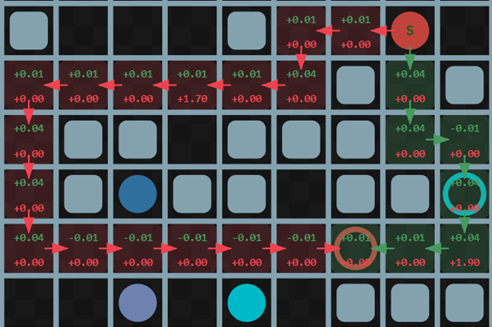}
    \caption{Two sample trajectories captured during an episode, along with the rewards the agent is expected to receive. 
    All neighboring agents remain stationary in this example while the red agent traverses toward its goal. 
    Green values indicate the learned reward (ARMS) while red values indicate the original sparse reward. 
    Since the agent has already partially progressed along its path, the sparse reward accumulated over the previous $20$ timesteps is received at timestep $5$. The rewards are scaled $\times 10$ for visual convenience.
    In the shorter trajectory, the agent reaches the goal exactly at timestep $7$, causing the accumulated sparse reward to be delivered two cells away from the goal. The sum of ARMS' reward at the shorter path is $+0.20$ while the longer path $+0.22$, indicating that higher reward can be accumulated when taking shorter paths.}
    \label{fig:two-trajectories}
\end{figure}

Finally, Figure~\ref{fig:two-trajectories} shows two example trajectories together with the rewards assigned along the path. 
Notably, the learned reward increases more sharply when the agent moves directly toward the goal, and encourages avoidance of other agents, and in this case receiving negative reward when it passes adjacently to them.

\FloatBarrier
\section{Conclusion and Future Work} \label{sec:future-work}

In this work, we presented ARMS, an \textbf{A}utomatic \textbf{R}eward-shaping framework for \textbf{M}ulti-agent \textbf{S}ystems which is able to integrate with any underlying multi-agent reinforcement learning algorithm. It aims to automatically infer rewards in sparse-setting environments by using the agent's sparse reward as a supervision signal to learn a parameterized \textit{dense} reward signal $\Rd_{\theta_i}$ for each agent $i$. 

The framework relies on the baseline algorithm to collect trajectories, and improves the parameterized reward signal by ranking trajectories based on the sparse reward. We theoretically prove that if certain assumptions hold, replacing \textbf{each} agent's reward function with another preserves the set of Nash-equilibria.

In our implementation the policy network parameters are shared across agents. Combined with the conditional best-response analysis, this symmetry motivates sharing the reward-shaping parameters and training the shaping network using observations aggregated from all agents, effectively treating each agent as an interchangeable instance of a representative learner interacting with copies of itself. 

Empirically, ARMS improves sampling efficiency in sparse and delayed reward settings and yields policies with higher cumulative reward. The resulting policies also complete tasks more frequently compared to the baseline. We further observe that the learned policies generalize better to unseen maps, outperforming the baseline across a diverse set of evaluation environments.

Future work may explore allowing the reward network to condition on agent histories. In the current implementation, the shaping function produces a reward at each timestep based only on the agent's local observation and action. While this may be sufficient if some temporal information is implicitly captured by the network parameters, incorporating sequence models that encode observation histories could allow the shaping function to capture temporal dependencies that are not observable from a single timestep.

Future work could also investigate the stability of the reward learning process, including the interaction between trajectory ranking,~reward~delay,~and~policy~optimization~dynamics.

\bibliography{main}
\bibliographystyle{tmlr}

\end{document}